\documentclass[num-refs]{wiley-article}

\usepackage{siunitx}
\usepackage{hyperref}
\hypersetup{colorlinks,linkcolor={red},citecolor={blue},urlcolor={red}} 
\papertype{Original Article}
\paperfield{Journal Section}

\title{Validation of Neural Network Controllers for Uncertain Systems Using the Keep-Close Approach: Robustness Analysis and Safety Verification}

\abbrevs{ABC, a black cat; DEF, doesn't ever fret; GHI, goes home immediately.}

\author[1\authfn{1}]{Abdelhafid Zenati}
\author[1\authfn{0}]{Nabil Aouf}

\contrib[\authfn{1}]{Equally contributing authors.}

\affil[1]{A. Zenati and N. Aouf are with the School of Mathematics, Computer Science and Engineering  Department of Electrical and Electronic Engineering, City, University of London.
}

\corraddress{ Tait Building, Northampton Square, London, EC1V 0HB. }
\corremail{abdelhafid.zenati@city.ac.uk}

\fundinginfo{Funder One, Funder One Department, Grant/Award Number: 123456, 123457 and 123458; Funder Two, Funder Two Department, Grant/Award Number: 123459}

\runningauthor{Zenati et al.}

\begin{document}

\begin{frontmatter}
\maketitle

\begin{abstract}
Ensuring the safety and robustness of neural network (NN) controllers in uncertain environments—characterised by unmodelled dynamics, nonlinearities, and time delays—remains a fundamental challenge in robust control systems. This paper introduces a novel method, termed "Keep-Close," for analysing the performance and robustness of uncertain feedback systems with NN controllers. The approach formulates the problem as tracking the dynamical error between an NN-controlled uncertain system and a robust reference model. The behaviour of the NN controller is characterised using the Differential Mean Value Theorem (DMV), Integral Quadratic Constraints (IQCs) and Linear Approximation (LA). A new dynamical system is formulated to describe the error, facilitating worst-case \(\mathcal{L}_2\) and \(\mathcal{L}_{\infty}\) error analysis through Lyapunov theory and IQCs. The proposed method is validated on two case studies—the Single-Link Robot Arm and the Apollo Lander—demonstrating its efficacy and versatility in assessing the robustness and performance of NN controllers in complex, uncertain environments.

\end{abstract}
\end{frontmatter}
\section{Introduction}
Neural Networks (NNs) show great performances when adopted for tasks where Artificial Intelligence (AI) based solutions are considered. Due to the latest advancements in AI, stabilising dynamical systems with NN controllers is invigorated \cite{ref1,ref0}. However,  feedback systems with NN controllers still suffer from a lack of robustness, safety certificates and,  validation due to their complex, opaque, and large-scale nature structures \cite{ref2, ref2b}.  Particularly, NN controllers have many types of nonlinear activation functions \cite{ref3} in addition to their vulnerability to adversarial attacks and system uncertainties \cite{ref4,ref5}. These drawbacks make the application of classical stability analysis methods such as Lyapunov theory to NN controller-based closed-loop systems difficult. This leads to limited adoption of NN controllers in safety-critical applications \cite{ref6,ref7}. 

Over the course of the last three decades,  the framework of uncertain system systematic analysis via Integral Quadratic Constraints (IQCs) was developed and used \cite{ref8,ref9}. The IQCs approach is proven to be efficient in capturing the properties of various classes of uncertainties such as slope-restricted nonlinearities, dynamics that can not be modeled,and time delays in addition to systems' nonlinearities \cite{ref10,ref11}. The flexibility of the approach and its generality made it one of the most important techniques utilized to perform and assess an uncertain dynamical system's stability analysis and performance, respectively \cite{ref12}. In \cite{ref13},  the problem of stability analysis of nonlinear uncertain systems is formulated using the definition of the augmented plant via the state-space factorization of the IQC.  The certificates for the region of attraction are established with both soft and hard IQC factorizations. Then,   the formulated problem is solved using a numerical approach such as Sum of Squares (SOS). Based on solving Linear Matrix Inequalities (LMIs), \cite{ref14} proposes a novel dynamic multipliers-dependent shift of the  Kalman-Yakubovich-Popov (KYP) \cite{ref15} certificate that is expressed using constraints.  The novel multipliers allow the formulation of the regional stability properties within IQCs theory, the verification of output constraints and the hard time-domain state. \cite{ref16} derives the certification conditions for local $\mathcal{L}_{2}$ gain of locally stable interconnected systems. Using estimation of reachability, the approach enhances the local $\mathcal{L}_2([0, \infty))$ gain calculation. Therefore, a set of local (IQCs) can be established for a fixed nonlinear dynamical system. Then by the local IQCs, a local $\mathcal{L}_2([0, \infty))$ gain is obtained. The possibility of using IQCs in the analysis systems subject to saturated linear feedback and the design of its control is the objective of the work in \cite{ref17}.  The research there focuses on establishing a IQCs-based conditions set under which an ellipsoid is contractively invariant for such a system. The paper results show that these set invariant conditions are necessary for such a study. 

 Several formal verification methods recently propose to investigate the stability of closed-loop systems interconnected with NN controllers. In \cite{ref18}, the authors study the verification problems for a class of piece-wise linear systems with NN controllers. By making an extension to the reachable set estimation, a reduction of the safety verification is established to check for empty junctions between unsafe regions and the reachable set. In \cite{ref19}, converting the NN controller into a hybrid system equivalent to the original one, the work addresses the issue of safety guarantees for autonomous systems with NN controllers. Moreover, the scalability of the proposed approach using Taylor series approximation with worst-case error bounds is looked at. Using a similar approach as \cite{ref19},  the work in \cite{ref20} emphasizes on NNs with both tanh and sigmoid activation functions. Through the Taylor model preconditioning and shrink wrapping,  a Taylor-model-based reachability algorithm is developed and a parallel implementation is provided. \cite{ref21} presents an approach, which permits the closed-loop system with NN controllers to involve different types of perturbations by using IQCs technique and to capture their input/output dynamical behaviour. 
 
 In this paper, we aim at certifying the safety and providing robustness guarantees for NN controllers against various systems uncertainties including unmodeled dynamics, slope-restricted non-linearities, and time delays. Inspired by the comparison principle \cite{comp,zenati}, the core idea of this work is to provide guarantees in maintaining the output of the uncertain system with a trained NN controller close to an ideal reference closed-loop model when its inputs change within a bounded set. To this end, a novel approach to analyse the dynamical error robustness of the feedback system with an NN controller and a reference closed-loop model is presented in this paper. Since analysing such systems is complicated, the problem is reformulated as the dynamical tracking errors between the uncertain interconnected system with an NN and the reference closed-loop model. Less conservative expression  of the neural controller error is then performed using the differential mean value theorem (DMV) and integral quadratic constraints (IQCs) in order to eliminate the neural controller's nonlinearity. As the time integral of the square of the error (ISE) is used to evaluate the system's performance \cite{oxford}, the bounded sets of the error between the outputs of the feedback system with an NN controller and a reference closed-loop model are derived based on generating dissipation-inequality conditions. These conditions of  bound the worst-case $\mathcal{L}_2\rightarrow \mathcal{L}_{2}$ and $\mathcal{L}_2\rightarrow \mathcal{L}_{\infty} $ induced gain of the dynamical error system (The Relative Integral Square Error (RISE) and the Supreme Square Error (SSE)), respectively. In addition, Lyapunov theory is integrated with IQCs-based techniques to accomplish the analysis of the worst-case gains. This analysis determines the worst-case ISE between the closed-loop reference model and the neural-controlled uncertain system. Using the difference between the closed-loop reference system and the interconnected system with a NN is an appealing approach for the following reasons: \textit{(i)} it can eliminate the nonlinearity of the NN controller by applying the Differential Mean Value theorem, which is more suitable and easier to analyse; \textit{(ii)} The work in \cite{ref30} considered the trivial initial conditions to facilitate the analysis. However, these conditions limit the application of their results to real challenging problems. Indeed, in the real world, the initial condition has a big effect on the stability of the system and generally, it is not null. Our keep-close technique can function on any initial condition of the system. Using the difference between the closed-loop reference system and the interconnected system with an NN controller, the initial conditions become null, which makes the analysis easier, plausible, and more general.  In \cite{ref2}, the authors analysed the stability of feedback systems with NN controllers using the IQCs approach to capture their input/output behaviour. Their theoretical result relies on semi-definite programming to estimate the Attraction Region (ROA) of the equilibrium point. Our keep-close approach is more general as it covers tracking in general where the convergence to an equilibrium point is a particular case. \textit{(iii)} Moreover, from a verification and safety point of view, the keep-close approach is more practical than the result in \cite{ref2} because it can indicate the occurrence likelihood of any undesirable behaviour before the convergence to the equilibrium point. The novel approach outlined in this paper has been implemented and tested in practice. Validation was carried out on two specific problems: the Single-Link Robot-Arm Control Problem, which focused on a single output scenario, and the Deep Guidance and Control of Apollo Lander, which dealt with multiple outputs. These tests provide empirical support for the efficacy and adaptability of the proposed approach across different application domains.

The rest of this paper is organized as follows: The problem
formulation and preliminaries are given in Section \ref{sec2}. Section \ref{sec3} is devoted to the robustness analysis when the system is submitted to uncertainties. 
In Section \ref{sec3}, a numerical example
is provided to illustrate the results of our proposed solution. Conclusions are given in
Section  \ref{sec4}. 
\subsection{\color{blue}Notation}
\newtheorem{deff}{Definition}

\begin{itemize}
    \item  $\mathbb{R}$, $\mathbb{R}_+$, and $\mathbb{C} $ refer to the set of real, non-negative real, and complex numbers, respectively. Also, in the complex plane, the notation $\mathbb{RH}$ denote to all analytic functions within the closed exterior of the unit disk. $M^T$ and $M^*$ are the transpose and complex conjugate transpose of the matrix $M$, respectively.
\item Consider the LPV system $\Sigma$ where its  matrices of state space depend on a time-varying
parameter vector $s(t): \mathbb{R}_+\rightarrow \mathcal{S} $ and $\mathcal{S} \subset \mathbb{R}^n$. $s(t)$ a continuously differentiable function of time where the parameter rates of variation $\dot{s}(t): \mathbb{R}_+\rightarrow \dot{\mathcal{S}}$. $\dot{\mathcal{S}}$ is a hyperrectangle given by: 
\begin{equation}\label{dots}
    \dot{\mathcal{S}}=\left\{ \rho \in \mathbb{R}^n~| ~
    \underline{\rho}_i \leq \rho_i \leq \overline{\rho}_i,~i=1\cdots n\right\}
\end{equation}
Let a parameter-dependent matrix $P$ be a continuously differentiable function of the parameter $s$, given by: $P: \mathcal{S}\rightarrow\mathbb{S}^n$, 
where $\mathbb{S}^n$ refers to the set of $n\times n$ symmetric matrices. The differential operator $\partial P: \mathcal{S} \times \dot{\mathcal{S}}\rightarrow \mathbb{S}^n$ is given by: 
\begin{equation}
    \partial P(s,\dot{s})=\sum\limits_{k=1}^n \dfrac{\partial P(s)}{\partial s_k}\dot{s}_k
\end{equation}
\item The space $\mathcal{L}_2([0, \infty))$ denotes the set of functions $\varphi: \mathbb{R}_+ \rightarrow \mathbb{R}^n$ satisfying $||\varphi||_2<\infty$
 \begin{eqnarray}
 ||\varphi||_2=  \left[\int\limits_{0}^{+\infty}\ \ \varphi(\tau)^T  \varphi(\tau)\,\mathrm{\textit{d} \tau}\right]^{\frac{1}{2}}
\end{eqnarray}
\item The space $\mathcal{L}_\infty$ refers to the set of functions $\varphi: \mathbb{R}_+ \rightarrow \mathbb{R}^n$ satisfying $||\varphi||_\infty<\infty$
 \begin{eqnarray}\label{iqctime}
 ||\varphi||_\infty=  \sup_{t\in [0,\infty)}\left[ \varphi(t)^T  \varphi(t)\right]^{\frac{1}{2}}
\end{eqnarray}
\item We recall here the Differential Mean Value (DMV) Theorem:

\begin{lemma}[\color{blue}DMV Theorem \cite{ref23,ref23b}]\label{theo1}
Consider a function \(g : \Omega \rightarrow \mathbb{R}^n\) that is differentiable, where \(\Omega\) is an open subset of \(\mathbb{R}^n\). Given two points \(a\) and \(b\) in \(\Omega\), such that the open line segment \(]a, b[\) lies entirely within \(\Omega\), there exists at least one point \(c\) on this segment for which
\begin{equation}
    g(b) - g(a) = \langle \nabla_x g(c), b - a \rangle,
\end{equation}
where \(c\) lies on the curve \(\varrho(\nu) = a + \nu(b - a)\) for some \(\nu \in [0, 1]\) called the convex domain $\mathbb{Co}(a,b)$. Here, \(\nabla_x g(c)\) denotes the gradient of \(g\) at \(c\), defined as
\begin{equation}
    \nabla_x g(c) = \left( \frac{\partial g}{\partial x_1}(c), \frac{\partial g}{\partial x_2}(c), \ldots, \frac{\partial g}{\partial x_n}(c) \right),
\end{equation}
representing the vector of partial derivatives of \(g\) with respect to each component of the input vector at the point \(c\).
\end{lemma}
 \item \begin{lemma}[Cauchy-Schwarz Inequality]
Let $\phi_1(t)$ and $\phi_2(t)$ be two square-integrable functions defined on the domain $[0, \infty)$. Then, the following inequality holds:
\[
\left| \int_0^\infty \phi_1(t) \phi_2(t) \, dt \right| \leq \sqrt{\int_0^\infty |\phi_1(t)|^2 \, dt} \cdot \sqrt{\int_0^\infty |\phi_2(t)|^2 \, dt}.
\]
\end{lemma}
\end{itemize}
\section{Problem formulation}\label{sec2}
 \subsection{\color{blue}Keep-Close Approach}
A robustness analysis of feedback systems utilising neural network (NN) controllers against potential uncertainties is presented herein. To achieve this objective, the output of the closed-loop system, equipped with an NN controller, is maintained closely to that of an ideal and trusted closed-loop reference model, especially when its input varies within a bounded set. Consider the closed-loop system, interconnected with an NN-based control system and subject to uncertainty \(\Delta_{\delta}\) as illustrated in Figure \ref{figplant}. This system, specifically the uncertain plant \(F(P,\Delta_{\delta})\) in conjunction with the NN controller \(\pi\), embodies the perturbed plant \(F(P,\Delta_{\delta})\). It represents an interconnection between a nominal plant \(P\) and an uncertainty \(\Delta_{\delta}\), which can be mathematically described by the following dynamical equations:
\begin{equation}\label{iqc2}
\Sigma_{F(P,\Delta_{\delta})}^{\pi}:=  \left\{  \begin{array}{l}\dot{x}(t)=Ax(t)+Bu(t)+\tilde{B}\delta(t)\\
    y(t)=Cx(t)+D u(t)+\tilde{D} \delta(t)\\
      \delta(t)=\Delta_{\delta} (y(t))\\
      u(t)=\pi(d(t), x(t)) \\
      x(0)=x_0
\end{array}\right.
\end{equation}
where $x \in \mathbb{R}^{n}$ is the system's state vector, $d \in \mathbb{R}^{d}$ is the input reference signal control, inputs of the neural controller $u\in \mathbb{R}^{m}$ and $(A,B,C,D, \tilde{B}, \tilde{D})$ are known with adequate dimensions, and the perturbation operator $\Delta_{\delta}$  is a bounded, causal. In this study,  we aim to ensure that the closed-loop system with a NN controller $\Sigma_{F(P,\Delta_{\delta})}^{\pi}$ keeps its outputs close to a  given reference model $\Sigma_{P}^{\pi^*}$ (a closed-loop system with an ideal, and classically proven controller). This reference model can be given by:
\begin{equation}\label{iqc1}
 \Sigma_{P}^{\pi^*}:=  \left\{   \begin{array}{l}\dot{\hat{x}}(t)=A\hat{x}(t)+B\hat{u}(t) \\
    \hat{y}(t)=C\hat{x}(t)+D\hat{u}(t)\\
    \hat{ u}(t)=\pi^*(d(t),\hat{x}(t)) \\
    \hat{ x}(0)=x_0
    \end{array}\right.
\end{equation}
with the reference system's state vector \(\hat{x} \in \mathbb{R}^{n}\), control inputs \(\hat{u} \in \mathbb{R}^{m}\), and the system matrices \((A, B, C, D)\) being the same as those mentioned in (\ref{iqc2}). After designing the controller \(\hat{u}(t) = \pi^*(d(t), \hat{x}(t))\) for the reference system \(\Sigma_{P}^{\pi^*}\), the system is transformed as follows:
\begin{equation}\label{refcon}
 \Sigma_{P}^{\pi^*}:=  \left\{   \begin{array}{l}\dot{\hat{x}}(t)=A_r\hat{x}(t)+B_r d(t) \\
    \hat{y}(t)=C_r\hat{x}(t)+D_r d(t)\\
    \hat{ x}(0)=x_0
    \end{array}\right.
\end{equation}
To accomplish our aim stated above, we present an analysis of the dynamical error system between the two closed-loop systems (i.e., the reference model and the system closed with an NN controller) in the following section. We start by mathematically describing the dynamical error system.

Using (\ref{iqc2}) and (\ref{iqc1}), let us assume that the state of the dynamic error system is defined as \(\zeta(t) = x(t) - \hat{x}(t)\) and \(z(t) = y(t) - \hat{y}(t)\). Consequently, the time derivative of the error state is given by \(\dot{\zeta}(t) = \dot{x}(t) - \dot{\hat{x}}(t)\), which leads to:
\begin{subequations}\label{erroriqc}
\begin{align}
\dot{\zeta}(t) &= A\zeta(t) + B\mu(t) + \tilde{B}\delta(t), \label{erroriqc-a} \\
z(t) &= C\zeta(t) + D\mu(t) + \tilde{D} \delta(t), \label{erroriqc-b} \\
\mu(t) &= \pi(d(t),x(t)) - \pi^*(d(t), \hat{x}(t)), \label{erroriqc-c} \\
\zeta(0) &= 0. \label{erroriqc-d}
\end{align}
\end{subequations}
where \(\mu(t)\) represents the controller error (between the NN controller and the ideal controller adopted in the closed-loop reference model). Furthermore, the initial condition of the dynamical error state \(\zeta(t)\) is \(\zeta(0) = 0\) because \(x(0) = \hat{x}(0) = x_0\).
\begin{figure}[thpb]
      \centering
   \includegraphics[scale=0.1]{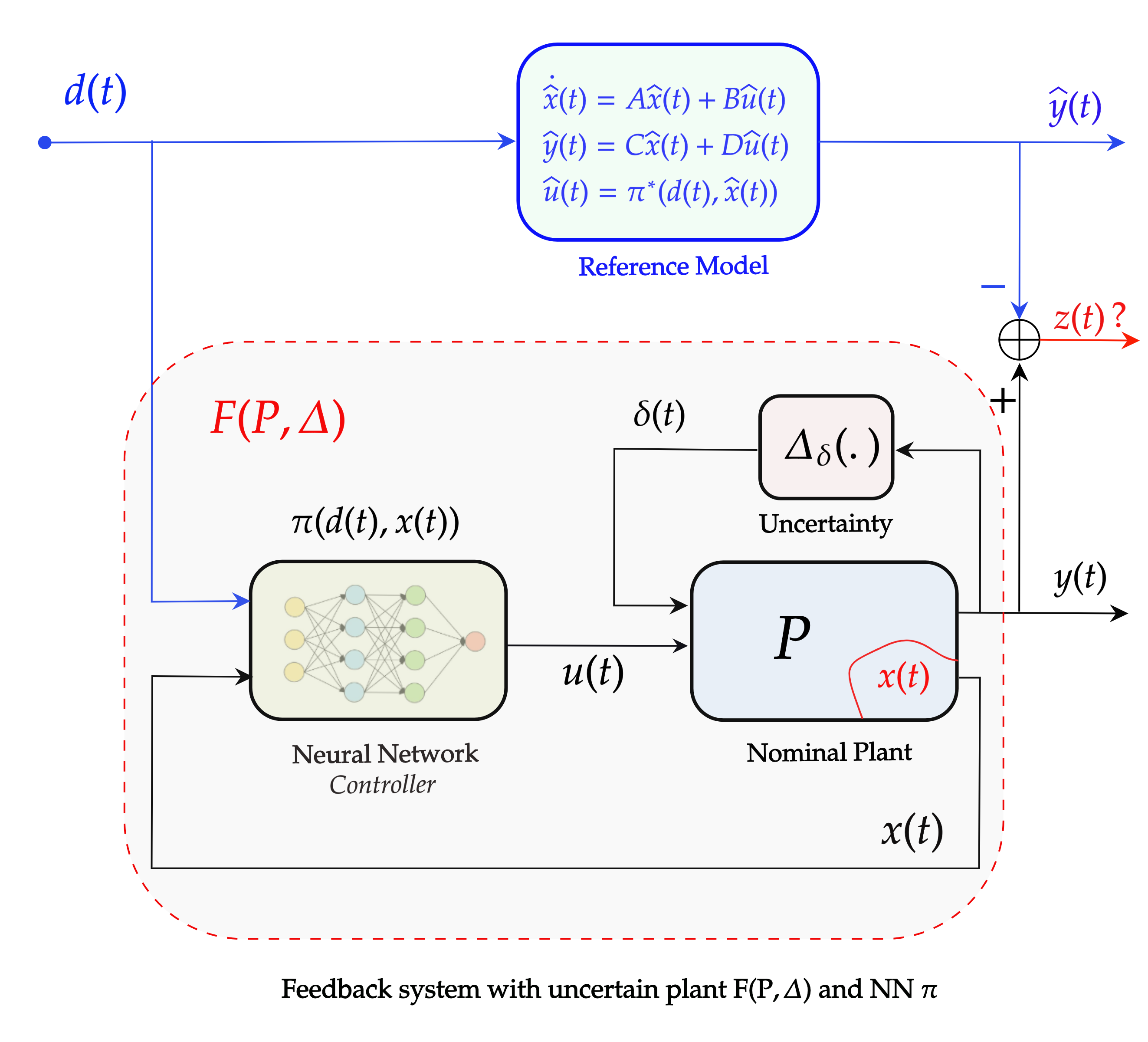}
\caption{Keep-Close Approach Illustration. }
      \label{figplant}
   \end{figure} 
 \subsection{\color{blue}Integral Quadratic Constraints (IQCs) Concept}
   The key concept of IQCs is instead of investigating the dynamical behaviour of the system that contains uncertainty $\Delta$, we analyse the system where $\Delta$ is removed and the signals $(p(t), ~q(t))$ are utilised to enforce the following constraints:
\begin{eqnarray}\label{iqc}
    \int\limits_{-\infty}^{+\infty}\ \ \begin{bmatrix}
 \hat{p}(j\omega)\\
 \hat{q}(j\omega)
\end{bmatrix}^* \Pi(j\omega) \begin{bmatrix}
 \hat{p}(j\omega)\\
 \hat{q}(j\omega)
\end{bmatrix}\,\mathrm{\textit{d} \tau}\geq0   
\end{eqnarray}
where $\Pi(j\omega)$ is named an “IQC multiplier” or can simply called a “multiplier”,  $\hat{p}(j\omega)$ and $\hat{q}(j\omega)$ are Fourier transforms of $p(t)$ and $q(t)$, respectively. Since (\ref{iqc}) is valid for all admissible choices of $\Delta$, then all properties  can  be proved for the constrained system are hold also for the original system. Therefore, the mathematical relationship between input-output quantities of system components can be described via IQCs.  The IQC in (\ref{iqc}) can be described in the time domain based on Parseval Theorem \cite{ref29a}. Assuming that $\Pi$ is function that is uniformly bounded and rational. Then,  $\Pi(j\omega)$  can be factorized as:   
\begin{equation}
    \Pi(j\omega)=\Psi(j\omega)^*\mathcal{M}\Psi(j\omega),
\end{equation}
 where $\Psi(j\omega)\in \mathbb{RH}$ and $\mathcal{M}$ represents a constant matrix \cite{ref11}.  $\Psi(j\omega)$ is expressed as:
 \begin{equation}
   \Psi(j\omega):=C_{\Psi}j[\omega I- A_{\Psi}][B_{\Psi1} ~B_{\Psi2}] + [D_{\Psi1}~ D_{\Psi2}]
\end{equation}
 Then, the IQC of (\ref{iqc}) can be written  as:
 \begin{eqnarray}\label{iqctime}
    \int\limits_{0}^{+\infty}\ \ r(\tau)^T \mathcal{M} r(\tau)\,\mathrm{\textit{d} \tau}\geq0
\end{eqnarray}
 where $r(t)$ is the output of the following linear system:
 \begin{equation}\label{iqcfilter}
   \Psi:= \left\{\begin{array}{ll}\dot{\psi}(t)&=A_{\Psi}\psi(t)+B_{\Psi1}p(t)+B_{\Psi2}q(t)\\
    r(t)&=C_{\Psi}\psi(t)+D_{\Psi1}p(t)+ D_{\Psi2}q(t)\\
      \psi(0)&=0
\end{array}\right.
\end{equation}
The time-domain constraint   (\ref{iqctime}) is generally satisfied over infinite time intervals. A hard IQC satisfies the following more restrictive conditions: If $\Delta$ is any causal and bounded operator that fulfills (\ref{iqc}) then:     
    \begin{eqnarray}
    \int\limits_{0}^{T}\ \ r(\tau)^T \mathcal{M} r(\tau)\,\mathrm{\textit{d} \tau}\geq0
\end{eqnarray}
 holds for all $T \geq 0$. Conversely, the soft IQCs in the time-domain constraint do not require to hold overall finite time intervals. This property is critical because the major technical and pedagogical problems in the IQCs framework occur when we utilize soft IQCs. Unfortunately, there is an ambiguity that surrounds both the terms soft IQC and hard IQC. In particular, the uniqueness of factorization of $\Pi(j\omega)$ as $  \Psi^*(j\omega)\mathcal{M}\Psi(j\omega)$, is not unique.   For clarity, the following definition will be adopted in the rest of the paper.
\begin{deff}[\color{blue} Hard IQC \cite{ref30}]
A uniformly bounded and rational function $\Pi : j\mathbb{R} \rightarrow \mathbb{C}$
admits a hard IQC factorization if there exists 
$\Psi(j\omega)\in \mathbb{RH}$ and matrix $\mathcal{M}$ with $\mathcal{M}_{ij}\in \mathbb{C}$ , such
that $\Pi(j\omega)= \Psi^*(j\omega)\mathcal{M}\Psi(j\omega)$, and any bounded, causal operator $\Delta$
which satisfies the IQC defined by $\Pi(j\omega)$ also satisfies:
\begin{eqnarray}\label{iqctime2}
    \int\limits_{0}^{T}\ \ r(\tau)^T \mathcal{M} r(\tau)\,\mathrm{\textit{d} \tau}\geq0
\end{eqnarray}
for all $T\geq0$ and for all $p(t)\in\mathcal{L}_2([0, \infty))([0~\infty))$, $q= \Delta(p)$. $(\Psi(j\omega), \mathcal{M} )$ is a hard IQCs factorization of $\Pi$.
\end{deff}
\subsection{\color{blue}Controller Error Analysis }
This subsection analyses the controller error \(\mu(t)\) in (\ref{erroriqc-c}), (i.e., the difference between the NN controller and the ideal controller adopted in the reference closed-loop model). The objective of this analysis is to identify a less conservative and more accurately representative expression for \(\mu(t)\), from which we can extract and dissociate certain and uncertain quantities. Achieving this will allow us to directly use this refined expression in the dynamic error system as referenced in (\ref{erroriqc}). To reach the aforementioned objective, we start by introducing the following lemma:
\begin{lemma}[Controller Error] \label{lemma4}
Consider the dynamics of the error as defined in (\ref{erroriqc}) and Let \(c\) the trajectory lies on the curve \(\vartheta (\nu) = x + \nu(\hat{x}- x)\) for some \(\nu \in [0, 1]\) called the convex domain $\mathbb{Co}(x,\hat{x})$. The controllers error \(\mu(t)\), as given in (\ref{erroriqc-c}), can be expressed as:
\begin{equation} \label{erroriqccont}
\mu(t)=\Lambda(s) \zeta(t)+\frac{\partial \alpha}{\partial \hat{x}}(0,0)\hat{x}(t)+\frac{\partial \alpha}{\partial d }(0,0) d(t)+\epsilon (t),
 \end{equation}
 where the trajecory $s(t)=[d(t), c(t)]^T$, the LPV matrix $ \Lambda(s(t))=\nabla_x \pi(d(t),c(t))$ is the Jacobian matrix 
 of the neural network, $\epsilon (t)=\mathcal{O}(d(t),\hat{x}(t))$ and  the training error  $\alpha(d(t),\hat{x}(t))=\pi(d(t), \hat{x}(t))-\pi^*(d(t), \hat{x}(t))$.
\end{lemma}
\begin{proof} From   (\ref{erroriqc-c}), the controller error $\mu(t)$ is given by:
\begin{equation}
\mu(t)=\pi(d(t), x(t))-\pi^*(d(t), \hat{x}(t))
 \end{equation}
Adding and subtracting $\pi(d(t),\hat{x}(t))$, the error approximation $\mu(t)$ can be expressed by:
\begin{equation}\label{adsu}
\begin{array}{ll}
 \mu(t)=\underbrace{\pi(d(t),\hat{x}(t))-\pi^*(d(t),\hat{x}(t))}_{=\alpha(d(t),\hat{x}(t))}+ \underbrace{\pi(d(t),x(t))-\pi(d(t),\hat{x}(t))}_{=\beta(t)}=\alpha(d(t),\hat{x}(t))+\beta(t)
\end{array}
 \end{equation}
It is clear that $\alpha(d(t),\hat{x}(t))$ is the training NN controller error  given by: 
\begin{equation}\label{epsilon}
 \alpha(d(t),\hat{x}(t))=\pi(d(t), \hat{x}(t))-\pi^*(d(t), \hat{x}(t))
 \end{equation}
 Knowing that $d(t),\hat{x}(t) \in \mathcal{L}_2([0, \infty))$  because it is the state of the reference model controlled by the ideal controller $\pi^*(d(t),\hat{x}(t))$. Therefore, by using Taylor approximation we can write the approximation of $\alpha(d(t),\hat{x}(t))$ as follows
\begin{equation}
\alpha(d(t), \hat{x}(t))= \alpha(0, 0)+\frac{\partial \alpha}{\partial \hat{x}}(0,0)\hat{x}(t)+\frac{\partial \alpha}{\partial r }(0,0) d(t)+\underbrace{\mathcal{O}(d(t),\hat{x}(t))}_{=\epsilon (t)}
\end{equation}
Clearly, if  $\hat{x}(t)=0$ and  $d(t)=0$ there is no controller action i.e. $\alpha(0,0)=0$, therefore
 \begin{equation}\label{eq19}
  \alpha(d(t), \hat{x}(t))= \frac{\partial \alpha}{\partial \hat{x}}(0,0)\hat{x}(t)+\frac{\partial \alpha}{\partial r }(0,0) d(t)+\epsilon (t)
\end{equation}
On the other hand, using the Differential Mean Value Theorem in Lemma \ref{theo1} and the fact that 
the  function $\pi$ is  continuous and differentiable  on convex hull of the set $\mathbb{Co}(x(t),\hat{x}(t))$, then, given two points \(x(t)\) and \(\hat{x}(t)\) in \(\Omega\) at time $t$, such that the open line segment  between \(x, \hat{x}\) lies entirely within \(\Omega\), there exists at least one point \(c(t)\) on this segment for which
\begin{equation}\label{eq18}
    \beta(t)= \pi(d(t),x(t))-\pi(d(t),\hat{x}(t)) = \langle \nabla_x \pi(d(t),c(t)), x(t) - \hat{x}(t) \rangle=\langle \nabla_x \pi(d(t),c(t)), \zeta(t) \rangle,
\end{equation}
where \(c\) lies on the curve \(\varrho(\nu) = x + \nu(\hat{x}- x)\) for some \(\nu \in [0, 1]\) called the convex domain $\mathbb{Co}(x,\hat{x})$. and 
\begin{equation}
     \nabla_x \pi(d(t),c(t)) = \left( \frac{\partial \pi}{\partial x_1}(d(t),c(t)), \frac{\partial \pi}{\partial x_2}(d(t),c(t)), \ldots, \frac{\partial \pi}{\partial x_n}(d(t),c(t)) \right),
\end{equation}
The equations (\ref{eq19}) and   (\ref{eq18})  allow us to conclude that:
 \begin{equation} \label{rev1}
\mu(t)=\underbrace{\langle \nabla_x \pi(d(t),c(t)), \zeta(t) \rangle}_{\Lambda(s) \zeta(t)}+\frac{\partial \alpha}{\partial \hat{x}}(0,0)\hat{x}(t)+\frac{\partial \alpha}{\partial r }(0,0) d(t)+\epsilon (t),
 \end{equation}
 where the trajectory $s(t)=[d(t), c(t)]^T$, the LPV matrix $ \Lambda(s(t))=\nabla_x \pi(d(t),c(t))$ is the Jacobian matrix 
 of the neural network, $\epsilon (t)=\mathcal{O}(d(t),\hat{x}(t))$ and  the training error  $\alpha(d(t),\hat{x}(t))=\pi(d(t), \hat{x}(t))-\pi^*(d(t), \hat{x}(t))$.
\end{proof}
The result in (\ref{rev1}) provides a less conservative and more representative expression for \(\mu(t)\). The specific quantity \(\Lambda(s(t)) \zeta(t)\) is influenced by the Jacobian of the neural controller $\nabla_x \pi(d(t),c(t))$, which incorporates the weights of the neural controller, and it is also dependent on the state of the dynamic error system \(\zeta(t)\). For uncertain components, the term \(\epsilon(t)\) is determined by the reference state \(\hat{x}(t)\) and the reference signal \(d(t)\), both of which belong to the \(\mathcal{L}_2([0, \infty))\) space.  In the next subsection, we will prove that \(\epsilon(t)\) also belongs to \(\mathcal{L}_2([0, \infty))\) and it is bounded, causal operator, leveraging this property.
\subsection{\color{blue}IQCs Characteristics of Uncertainties}
The perturbations acting on the feedback system can encompass various types of uncertainties, including unmodeled dynamics, saturation effects, time delays, and slope-restricted nonlinearities. In the next lemma, we will demonstrate that if \(d(t)\) and \(\hat{x}(t)\) belong to \(\mathcal{L}_2([0, \infty))\), then \(\epsilon(t) = \mathcal{O}(d(t), \hat{x}(t))=\Delta_\epsilon (\eta(t))\) where $\eta(t) =[d(t), \hat{x}(t)]^T$ also belongs to \(\mathcal{L}_2([0, \infty))\) and it is bounded and causal operator. Establishing this result will enable us to apply the IQCs technique to address the uncertainty \(\epsilon(t)\) in the subsequent analysis.
\begin{lemma}
Let  $\eta(t) =[d(t), \hat{x}(t)]^T$ belongs to \(\mathcal{L}_2([0, \infty))\). Then, the quantity \( \epsilon(t) =\Delta_{\epsilon} (\eta(t)) = \mathcal{O}(d(t), \hat{x}(t))\) also belongs to \(\mathcal{L}_2([0, \infty))\) and it is bounded and causal operator. 
\end{lemma}
\begin{proof}
By assumption, \(d(t)\) and \(\hat{x}(t)\) belong to the space \(\mathcal{L}_2([0, \infty))\). Additionally, we have \(\epsilon(t) = \mathcal{O}(d(t), \hat{x}(t))\), which implies that there exist positive constants \(\kappa_1 \in ]0, 1[\) and \(\kappa_2 \in ]0, 1[\) such that:
\begin{equation}\label{tho}
    |\epsilon(t)| = |\mathcal{O}(d(t), \hat{x}(t))| \leq \kappa_1 |d(t)| + \kappa_2 |\hat{x}(t)|.
\end{equation}
This implies:
\begin{equation}
    \int\limits_{0}^{+\infty} |\epsilon(t)|^2 \,\mathrm{d}t \leq \int\limits_{0}^{+\infty} \left(\kappa_1 |d(t)| + \kappa_2 |\hat{x}(t)|\right)^2 \,\mathrm{d}t.
\end{equation}

Expanding the squared term and applying the triangle inequality under the integral gives:
\begin{align}
    \int\limits_{0}^{+\infty} |\epsilon(t)|^2 \,\mathrm{d}t &\leq \int\limits_{0}^{+\infty} \left( \kappa_1^2 |d(t)|^2 + \kappa_2^2 |\hat{x}(t)|^2 + 2\kappa_1\kappa_2 |d(t)||\hat{x}(t)| \right) \,\mathrm{d}t. 
\end{align}
Since \(d(t)\), \(\hat{x}(t) \in \mathcal{L}_2([0, \infty))\), their respective norms are finite:
\begin{equation}
    \|d(t)\|_{\mathcal{L}_2}^2 = \int\limits_{0}^{+\infty} |d(t)|^2 \,\mathrm{d}t < \infty, \quad \|\hat{x}(t)\|_{\mathcal{L}_2}^2 = \int\limits_{0}^{+\infty} |\hat{x}(t)|^2 \,\mathrm{d}t < \infty.
\end{equation}
The cross-term \(2\kappa_1\kappa_2 |d(t)||\hat{x}(t)|\) is also integrable because of the Cauchy-Schwarz inequality:
\begin{equation}
    \int\limits_{0}^{+\infty} |d(t)||\hat{x}(t)| \,\mathrm{d}t \leq \|d(t)\|_{\mathcal{L}_2} \|\hat{x}(t)\|_{\mathcal{L}_2}.
\end{equation}
Thus, \(\epsilon(t) \in \mathcal{L}_2([0, \infty))\) and it is bounded operator. Also, it is clear that the output of $\epsilon(t)$ at any time $t$ depends only on the input $(d(t), \hat{x}(t))$ up to that time $t$ i.e $\epsilon(t)$ is causal operator. This completes the proof and then the IQCs technique can be applied  to address the uncertainty \(\epsilon(t)=\Delta_\epsilon (\eta(t))\).
\end{proof}
 Moreover, the input-output relationship of the uncertainties, denoted by \(\delta(t)=\Delta_{\delta}(y (t))\), can be mathematically described using IQCs. Thus, the IQC technique can be employed to characterise the uncertainties \(\epsilon(\cdot)\) and \(\Delta_{\delta}(\cdot)\), leading to the following definition:
\newtheorem{assum}{Assumption}
\begin{assum}\label{assumption1}
Let  Assume that $y(t)$ and $\eta(t)=[d(t) \hat{x}(t) ]^T$ belong to \(\mathcal{L}_2([0, \infty))\), and are governed by bounded, causal operators \(\Delta_{\delta}(\cdot)\) and \(\Delta_{\epsilon}(\cdot)\), respectively. These operators satisfy the Integral Quadratic Constraints (IQCs) defined by   \(\Pi_{{\delta}}(j\omega)\) for \(\delta(t)\) and \(\Pi_{\epsilon}(j\omega)\) for \(\epsilon(t)\), respectively. The input-output relationship of the uncertainties, \(\Delta_{\delta}(\cdot)\) and \(\epsilon(\cdot)\), can be mathematically described by IQCs as follows:
\begin{equation}\label{iqctime3}
\begin{aligned}
    &\int_{0}^{T} r_{{\delta}}(\tau)^T \mathcal{M}_{{\delta}} r_{{\delta}}(\tau) \, d\tau \geq 0, ~~
    &\int_{0}^{T} r_{\epsilon}(\tau)^T \mathcal{M}_{\epsilon} r_{\epsilon}(\tau) \, d\tau \geq 0,
\end{aligned}
\end{equation}
for all \(T \geq 0\) and for all signals \( y(t), \eta(t) \in \mathcal{L}_2([0, \infty))\). Here, \((\Psi_{{\delta}}(j\omega), \mathcal{M}_{{\delta}})\)  and \((\Psi_{\epsilon}(j\omega), \mathcal{M}_{\epsilon})\) represent a hard IQC factorization of \(\Pi_\delta\) and  \(\Pi_{\epsilon}\), respectively.
\end{assum}
\begin{corollary}\label{corollaryExample}
Given Assumption \ref{assumption1}, the signals \(r_{\delta}(t)\) and \(r_\epsilon(t)\), as defined in (\ref{iqctime3}), are the outputs of the following linear time-invariant (LTI) systems, respectively:
\begin{equation}\label{delta}
\begin{cases}
\dot{\psi}_{\delta}(t) = A_{\delta}\psi_{\delta}(t) + B_{\delta 1}\delta(t) + B_{\delta 2}y(t), \\
r_{\delta}(t) = C_{\delta}\psi_\delta(t) + D_{\delta 1}\delta(t) + D_{\delta 2}y(t), \\
\psi_{\delta}(0) = 0,
\end{cases}
\end{equation}
and
\begin{equation}\label{epsilon}
\begin{cases}
\dot{\psi}_\epsilon(t) = A_{\epsilon}\psi_\epsilon(t) + B_{\epsilon1}\epsilon(t) + B_{\epsilon2}\eta(t), \\
r_\epsilon(t) = C_{\epsilon}\psi_\epsilon(t) + D_{\epsilon1} \epsilon(t) + D_{\epsilon2}\eta(t), \\
\psi_\epsilon(0) = 0,
\end{cases}
\end{equation}
\end{corollary}
 where $\Psi_{\delta}$ and $\Psi_\epsilon$ represent the ‘virtual’ filters as in  (\ref{iqcfilter}) applied to the inputs $y$ and $\eta$ and the outputs $\Delta_{\delta}$ and $\epsilon$ of $\Delta_{\delta}$ and $\epsilon$ and a constraint on the outputs $r_{\delta}$ and $r_\epsilon$  of  (\ref{iqcfilter}). Thus, an extended filter can be written as follows: 
   \begin{equation}\label{filter}
   \Psi:= \left\{\begin{array}{ll}\dot{\xi}(t)&=A_{\xi}\xi(t)+B_{\xi1}q(t)+B_{\xi2}p(t)\\
    r(t)&=C_{\xi}\xi(t)+D_{\xi1}q(t)+ D_{\xi2}p(t)\\
      \xi(0)&=0
\end{array}\right.
\end{equation}
where \begin{eqnarray*}\xi(t)=  \begin{bmatrix}
          \psi_{\delta}(t)\\\psi_\epsilon(t)
         \end{bmatrix},~~
 p(t)=  \begin{bmatrix}
          y(t)\\\eta(t)
         \end{bmatrix},~q(t)=  \begin{bmatrix}
          \delta(t)\\\epsilon(t)
         \end{bmatrix}
\end{eqnarray*}
and the matrices of  (\ref{filter}) are defined as follow:
\begin{eqnarray*}
 A_{\xi}=  \begin{bmatrix} A_{{\delta}}& 0\\ 0& A_{\epsilon}\end{bmatrix},  B_{\xi1}=  \begin{bmatrix} B_{{\delta}1}& 0\\ 0& B_{\epsilon1}\end{bmatrix},~  B_{\xi2}=  \begin{bmatrix} B_{{\delta}2}& 0\\ 0& B_{\epsilon2}\end{bmatrix}\\
 C_{\xi}=  \begin{bmatrix} C_{{\delta}}& 0\\ 0& C_{\epsilon}\end{bmatrix},  D_{\xi1}=  \begin{bmatrix} D_{\delta1}& 0\\ 0& D_{\epsilon1}\end{bmatrix},~  D_{\xi2}=  \begin{bmatrix} D_{\delta2}& 0\\ 0& D_{\epsilon2}\end{bmatrix}
\end{eqnarray*}
 \section{ Robust Tracking Analysis via Standard IQCs}\label{sec3}
 Considering the error dynamics system in (\ref{erroriqc}) and the conservative expression of \(\mu\) presented in (\ref{erroriqccont}), let us define \(\chi(t)\) as the extended state vector, given by
\[
\chi(t) = \begin{bmatrix}
          \zeta(t) \\
          \xi(t)
         \end{bmatrix},
\]
where \(\zeta(t)\) represents the system state error and \(\xi(t)\) encapsulates additional dynamics or states relevant to the system analysis. The dynamics of \(\chi(t)\), omitting the dependency on \(s\) for brevity, are governed by the following equations:
          \begin{equation}\label{extendedsystem}
  \Sigma_{ex}\left\{\begin{array}{ll}\dot{\chi}(t)&=\mathcal{A}\chi(t)+\mathcal{B}_1
  q(t)+\mathcal{B}_2 \eta(t)\\
   r(t)&=\mathcal{C}_1\chi(t)+\mathcal{D}_{11}q(t)+ \mathcal{D}_{12} \eta(t)\\
z(t)&=\mathcal{C}_2\chi(t)+\mathcal{D}_{21}q(t)+ \mathcal{D}_{22} \eta (t)\\
   
         \eta(t)&=\begin{bmatrix}
          d(t)& \hat{x}(t)
         \end{bmatrix}^T, ~\chi(0)=0
\end{array}\right.
\end{equation}
where the matrices of the Linear Parameter Varying (LPV) extended system  (\ref{extendedsystem}) have the appropriate dimensions.
  \begin{equation}\begin{array}{l}
    \mathcal{A}(s)=  \begin{bmatrix}A + B\Lambda(s)&0&0\\ B_{\delta 2}\Big(C + D\Lambda(s)\Big) & A_{\delta}&0\\
   0&0&  A_{\epsilon}\end{bmatrix}
    \\
    \mathcal{B}_1(s)= \begin{bmatrix}\tilde{B}&B\\ B_{\delta 2}\tilde{D} & B_{{\delta}2}D\\ D_{{\delta}2}\tilde{D} & D_{{\delta}2}D\end{bmatrix}, ~\mathcal{B}_2=\begin{bmatrix}\begin{bmatrix}
B\frac{\partial \alpha}{\partial d}(0,0) & B\frac{\partial \alpha}{\partial \hat{x}}(0,0)
\end{bmatrix}\\ \begin{bmatrix}
B_{{\delta}2}D\frac{\partial \alpha}{\partial d}(0,0) + B_{{\delta}2}D_r & B_{{\delta}2}D\frac{\partial \alpha}{\partial \hat{x}}(0,0) + B_{{\delta}2}C_r
\end{bmatrix}\\B_{\epsilon2} \end{bmatrix} \\
    \mathcal{C}_1(s)=  \begin{bmatrix}D_{{\delta}2}\Big(C + D\Lambda(s)\Big)\\  C_{{\delta}}\\C_{\epsilon}\end{bmatrix}, ~ \mathcal{C}_2(s)=  \begin{bmatrix}C + D\Lambda(s) \\ 0\\ 0\end{bmatrix} \\ \mathcal{D}_{11}= \begin{bmatrix}
D_{{\delta}2}\tilde{D}+D_{{\delta}1} & D_{{\delta}2}D\\
0& D_{\epsilon1}
\end{bmatrix}, ~\mathcal{D}_{12}=\begin{bmatrix}\begin{bmatrix}
D_{{\delta}2}D\frac{\partial \alpha}{\partial d}(0,0) + D_{{\delta}2}D_r & D_{{\delta}2}D\frac{\partial \alpha}{\partial \hat{x}}(0,0) + D_{{\delta}2}C_r 
\end{bmatrix}\\  D_{\epsilon2}\end{bmatrix} \\~\mathcal{D}_{21}=\begin{bmatrix}
\tilde{D} & D
\end{bmatrix},~\mathcal{D}_{22}(s)= \begin{bmatrix}
D\frac{\partial \alpha}{\partial d}(0,0) & D\frac{\partial \alpha}{\partial \hat{x}}(0,0)
\end{bmatrix}
    \end{array}
\end{equation}
This extended system is derived from the dynamics of the error system as detailed in the Appendix \ref{appendix}, the approximation error \(\mu(t)\) expression given in (\ref{erroriqccont}), and the extended virtual filter as described in (\ref{filter}). According to \cite{ref30}, the resulting system, delineated in (\ref{extendedsystem}), is depicted in Figure \ref{eqiv}. This figure illustrates a system that accepts the state of the uncertain system \(\eta(t)\) as input and produces the error output \(z(t)\), incorporating uncertainties \(\Delta_{\delta}\) and \(\Delta_\epsilon\), as shown in Figure \ref{eqiv}.

On the other hand, a dissipation inequality can be formulated to upper bound the worst-case \(\mathcal{L}_2([0, \infty))\) gain of \(\Tilde{F}(\Tilde{P},\Delta_{\delta},\Delta_\epsilon)\) as represented in Figure \ref{eqiv}. This is achieved by utilising the extended system described in (\ref{extendedsystem}) and the time-domain Integral Quadratic Constraints (IQCs) specified in (\ref{iqctime3}).
\begin{figure}[thpb]
      \centering
   \includegraphics[scale=0.1]{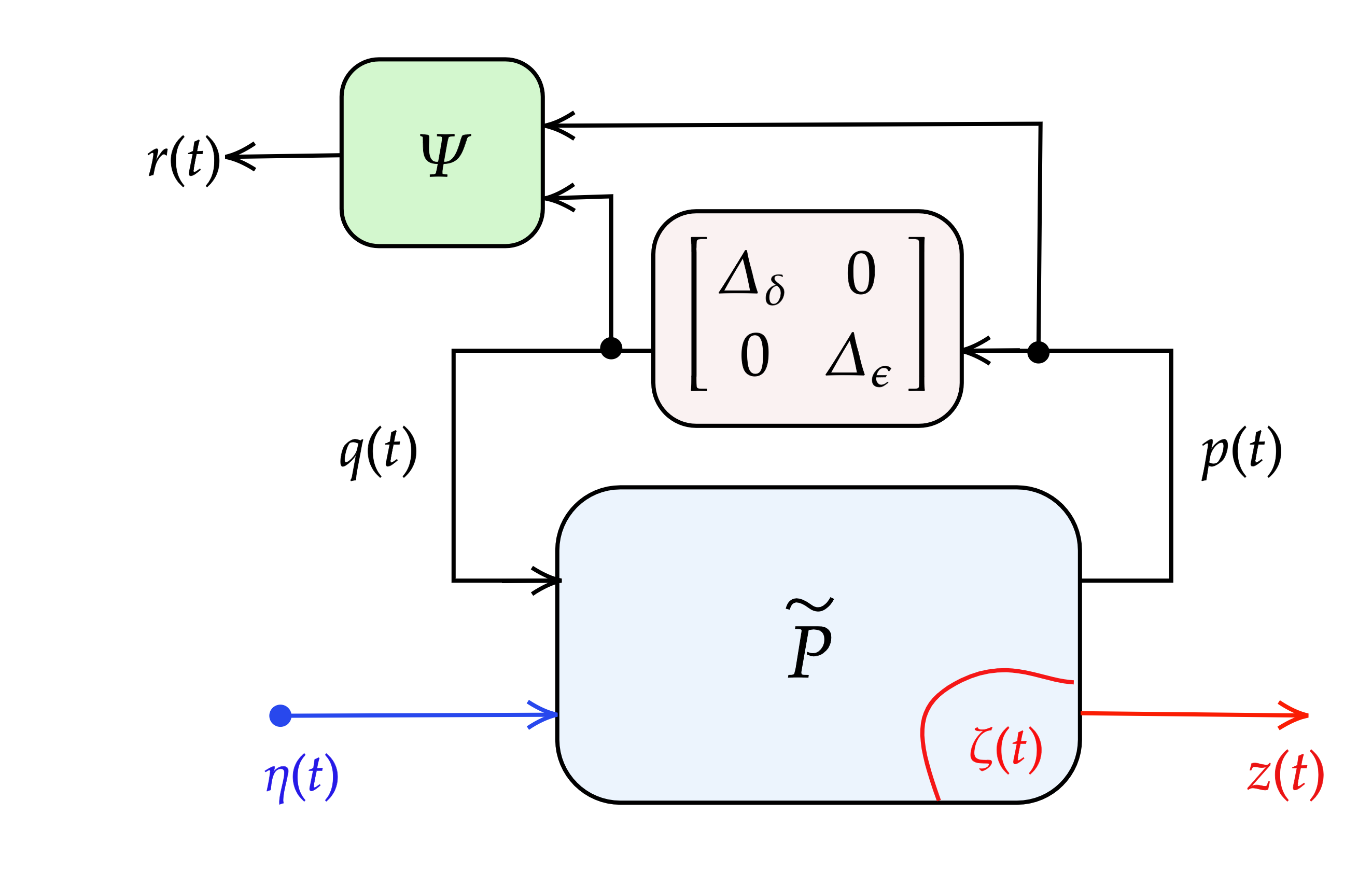}
\caption{ Graphical Interpretation of the IQCs with the Equivalent System of (\ref{extendedsystem}). }
\label{eqiv}
   \end{figure} 
As the aim in this work is to estimate the Integral Square Error (ISE) and worst expected Supreme Square Error (SSE), it is worth to define those metrics as following:
   \begin{deff}[\color{blue} RISE \& SSE]\label{difer}
       Let us consider $y$ and $\hat{y}$ are the outputs of the closed-loop interconnected-based NN system and the closed-loop reference model, respectively. Then  the Relative Integral Square Error (RISE) and the Supreme Square Error (SSE) are given by :
       \begin{equation}
           RISE\triangleq \dfrac{||y-\hat{y}||_2}{\sqrt{||d||_2^2  + ||\hat{x}||_2^2}}, ~SSE\triangleq \dfrac{||y-\hat{y}||_\infty}{\sqrt{||d||_2^2  + ||\hat{x}||_2^2}}
       \end{equation}
       where the worst case of the Relative Integral Square Error $\gamma$ and the worst case of the Supreme Square Error $\sigma$ satisfy:
       \begin{equation}
        \dfrac{||y-\hat{y}||_2}{\sqrt{||d||_2^2  + ||\hat{x}||_2^2}}\leq \gamma,~ ~\dfrac{||y-\hat{y}||_\infty}{\sqrt{||d||_2^2  + ||\hat{x}||_2^2}}\leq \sigma
       \end{equation}
   \end{deff}
\subsection{\color{blue} Relative Integral Square Error (RISE) }
The following theorem delineates the conditions necessary to evaluate the worst-case performance of the Relative Integral of the Square of the Error (RISE) metric. In control theory, the RISE metric is instrumental in assessing the efficacy of a system's response \cite{oxford}. Specifically, RISE offers critical insights into the comparative behaviour of an uncertain closed-loop system, equipped with a neural network (NN) controller, against a predefined closed-loop reference model. This comparison is pivotal for understanding the resilience and robustness of the NN-controlled system in the face of uncertainties.
\newtheorem{theo}{Theorem}
\begin{theo}[Worst-case RISE Performance]
Let \(\Delta_{\delta}\) and \(\Delta_\epsilon\) satisfy the Integral Quadratic Constraint (IQC) given by \(IQC(\Xi, M)\), and assume that the extended system \(\Sigma_{ex}\) as defined in (\ref{extendedsystem}) is well-posed. Suppose there exists a scalar \(\lambda > 0\) and a continuously differentiable matrix \(P = P^T\) with appropriate dimensions such that \(P \geq 0\), and for all \((s, \dot{s}) \in \Omega \times \dot{\mathcal{S}}\) (omitting \(s\) and \(\dot{s}\) for brevity):
\begin{equation}\label{h}
\begin{bmatrix}
P\mathcal{A} + \mathcal{A}^T P + \partial P & P\mathcal{B}_1 & P\mathcal{B}_2 \\
\mathcal{B}_1^T P & 0 & 0 \\
\mathcal{B}_2^T P & 0 & -I
\end{bmatrix} + \frac{1}{\gamma^2}
\begin{bmatrix} \mathcal{C}_2^T \\ \mathcal{D}_{21}^T \\ \mathcal{D}_{22}^T \end{bmatrix}
\begin{bmatrix} \mathcal{C}_2 & \mathcal{D}_{21} & \mathcal{D}_{22} \end{bmatrix} + \lambda
\begin{bmatrix} \mathcal{C}_1^T \\ \mathcal{D}_{11}^T \\ \mathcal{D}_{12}^T \end{bmatrix}
M
\begin{bmatrix} \mathcal{C}_1 & \mathcal{D}_{11} & \mathcal{D}_{12} \end{bmatrix} < 0
\end{equation}
Then, the Relative Integral Square Error (RISE) satisfies:
\begin{equation}
\sup_{s \in \Omega, 0 \neq \eta \in \mathcal{L}_2, \zeta(0) = 0} \left[ \frac{||y - \hat{y}||_2}{\sqrt{||d||_2^2 + ||\hat{x}||_2^2}} \right] \leq \gamma,
\end{equation}
\end{theo}
\begin{proof}
Consider the various signals \((\chi, z, p, \eta, r)\) of the error dynamics system, with the input \(\eta \in \mathcal{L}_2([0, \infty))\) and the solution trajectory \(\chi\) subject to zero initial conditions \(\chi(0) = 0\). Assuming the extended system, as referred to in (\ref{extendedsystem}), is well-posed implies that all signals within this system are well-defined. By assumption, the uncertainties \(\Delta_{\delta}\) and \(\Delta_\epsilon\) are constrained by the Integral Quadratic Constraint (IQC) defined by \((\Xi, M)\), and consequently, the quantity \(r\) must satisfy the time-domain IQC as given in (\ref{iqctime3}) for any \(T > 0\). Furthermore, utilizing the state-space representation of the virtual filter detailed in (\ref{iqcfilter}), the expression for \(r^T M r\) is formulated as follows:
\begin{equation}\label{str}
    rMr^T= \begin{bmatrix}\underline{\xi}\\p\\q \end{bmatrix}^T\begin{bmatrix}C_\xi^T M C_\xi& C_\xi^T M D_\xi\\
  D_\xi^T M C_\xi& D_\xi^T M D_\xi  
    \end{bmatrix} \begin{bmatrix}\underline{\xi}\\p\\q \end{bmatrix}
\end{equation}
The foundation of the following proof lies in defining a parameter-dependent storage function, also known as a Lyapunov function, \(V: \mathbb{R}^\chi \times \Omega \rightarrow \mathbb{R}_+\), given by \(V(\chi,s) = \chi^T P(s) \chi\). The strict inequality presented in (\ref{h}) suggests that there exists a very small quantity \(\varepsilon > 0\) such that the following perturbed matrix inequality is satisfied for all \(s \in \Omega\):
\begin{eqnarray}\label{L2}\begin{array}{l}
\begin{bmatrix}
          P\mathcal{A}^T+\mathcal{A}P^T+\partial P& P\mathcal{B}_1  &P\mathcal{B}_2 \\
           \mathcal{B}_1^T P &0&0\\
           \mathcal{B}_2^T P&0&-(1-\varepsilon)I
         \end{bmatrix}+\dfrac{1}{\gamma^2}\begin{bmatrix} \mathcal{C}_2^T\\\mathcal{D}_{21}^T\\\mathcal{D}_{22}^T\end{bmatrix}\begin{bmatrix} \mathcal{C}_2&\mathcal{D}_{21}&\mathcal{D}_{22}\end{bmatrix}+ \lambda\begin{bmatrix} \mathcal{C}_1^T\\\mathcal{D}_{11}^T\\\mathcal{D}_{12}^T\end{bmatrix}M\begin{bmatrix} \mathcal{C}_1&\mathcal{D}_{11}&\mathcal{D}_{12}\end{bmatrix}\leq0
         \end{array}
\end{eqnarray}
Taking into account the structural condition given in (\ref{str}) and multiplying both sides of equation (\ref{L2}) from the left and right by \([\chi^T, q^T, \eta^T]\) and its transpose \([\chi^T, q^T, \eta^T]^T\) respectively, we observe that the Lyapunov function \(V\) satisfies the following dissipation inequality:
\begin{eqnarray}\label{diss}
   \nabla_{{\chi}} V \dot{\chi}+\nabla_{{s}} V\dot{s}\leq-\dfrac{1}{\gamma^2} z^Tz+(1-\varepsilon)\eta^T\eta-\lambda rMr^T
\end{eqnarray}
This analysis demonstrates that the dissipation inequality presented in (\ref{diss}) corresponds directly to the perturbed linear matrix inequality delineated in (\ref{h}). By integrating the dissipation inequality obtained from (\ref{diss}) over the interval from the lower bound time \(t = 0\) to the upper bound time \(t = T\), and considering the initial condition \(\chi(0) = 0\) i.e $V(0)=0$, we obtain:
\begin{eqnarray}\label{diss2}
    \underbrace{V(\chi(T), s(T))-V(0,0)}_{=V(\chi(T), s(T))\geq0~ \because ~V(0,0)=0   }+\lambda  \int\limits_{0}^{T}\ \ r(\tau)^T \mathcal{M} r(\tau)\,\mathrm{\textit{d} \tau}\leq (1-\varepsilon)\int\limits_{0}^{T}\ \ \eta(\tau)\eta(\tau)^T\,\mathrm{\textit{d} \tau}-\dfrac{1}{\gamma^2}\int\limits_{0}^{T}\ \ z(\tau)z(\tau)^T\,\mathrm{\textit{d} \tau}
\end{eqnarray}
It follows from the positivity condition  of IQCs as in (\ref{iqctime3}), $\lambda \geq 0$, and also due to the non-negativity of the storage function $V$ leads to:
\begin{equation}
    \dfrac{1}{\gamma^2}\ \int\limits_{0}^{T}\ \ z(\tau)z(\tau)^T\,\mathrm{\textit{d} \tau} \leq (1-\varepsilon)\int\limits_{0}^{T}\ \ \eta(\tau)\eta(\tau)^T\,\mathrm{\textit{d} \tau}
\end{equation}
As we assume that $\varepsilon>0$ , thus
\begin{eqnarray}
    \lim_{T\rightarrow +\infty}\left[ \dfrac{1}{\gamma^2}\ \int\limits_{0}^{T}\ \ z(\tau)z(\tau)^T\,\mathrm{\textit{d} \tau} < \int\limits_{0}^{T}\ \ \eta(\tau)\eta(\tau)^T\,\mathrm{\textit{d} \tau}\right]
\end{eqnarray}
Knowing that 
\begin{equation}\begin{array}{l}
\left|\left|z(t)\right|\right|_2=||y(t)-\hat{y}(t)||_2 \,~~ \left|\left|\eta(t)\right|\right|_2=  \sqrt{||d(t)||_2^2  + ||\hat{x}(t)||_2^2}
    \end{array}
\end{equation}
this allows us to  deduce  that:
 \begin{equation}\label{eq38}
\sup_{s \in \Omega.}\sup_{0\neq y, \hat{y}\in \mathcal{L}_2, \zeta(0)=0}  \dfrac{||y-\hat{y}||_2}{\sqrt{||d||_2^2  + ||\hat{x}||_2^2}}\leq\gamma
 \end{equation}
 \end{proof}

The value \(\gamma\) specified in Theorem \ref{theo1} represents the maximal \(\mathcal{L}_2\) gain across all permissible parameter trajectories. Specifically, it defines the upper bound on the worst-case Relative Integral of the Square of the Error (RISE) between the closed-loop reference model and the NN controlled system. This worst-case scenario accounts for the NN controller parameters \(\Lambda\) and $\alpha$ and their associated training error \(\epsilon\), providing crucial insights into the system's behaviour under the influence of uncertainties \(\Delta_{\delta}\) and \(\Delta_{\epsilon}\). 

The result of Theorem \ref{theo1} is particularly significant, as it implies that, by relying solely on the known quantities \(d(t)\) and \(\hat{x}(t)\), we can predict the behaviour of the uncertain system controlled by the NN. More specifically, it allows us to determine the expected regions and boundaries within which the system’s output will remain, despite uncertainties and nonlinearities.

This insight is particularly crucial in control applications, as it facilitates the assessment of the robustness and reliability of the neural network controller in maintaining stability and performance under varying conditions. By characterising these boundaries, we gain a deeper understanding of the extent to which the trained neural controller approximates the desired closed-loop reference model. Moreover, this formulation provides a foundation for further optimisation, enabling the refinement of the neural controller to enhance its adaptability and accuracy in real-world applications. The worst-case result can be leveraged to improve the robustness of the NN controller by minimising the dynamical error, thereby guaranteeing better performance when implemented. Additionally, this worst-case measure serves as an indicator of the likelihood of undesirable behaviour emerging within the closed-loop system when employing the designed NN controller.

Such insights are invaluable for enhancing the robustness of NN controllers by reducing the dynamical error, ultimately ensuring superior performance upon implementation. Furthermore, this worst-case performance measure serves as a predictive tool for assessing the potential occurrence of adverse behaviours within the closed-loop system equipped with the NN controller. This predictive capability is instrumental in preemptively addressing and mitigating risks associated with the deployment of NN-based control strategies, highlighting the theorem’s practical significance in the design and optimisation of robust NN controllers.

\subsection{\color{blue}Supreme Square Error (SSE) }
The Linear Parameter Varying (LPV) model approximation, employed in the robust tracking problem elaboration, invites further scrutiny under more conservative conditions. Specifically, the robust performance analysis of the LPV model for system (\ref{extendedsystem}) can be broadened to encompass performance metrics less conservative than the Relative Integral Square Error (RISE). The forthcoming theorem aims to establish the conditions for the worst-case Supreme Square Error (SSE), denoted by \(\sigma\), under conditions that are less conservative, such as \(||y-\hat{y}||_2 \geq ||y-\hat{y}||_\infty\). More granular than RISE, SSE offers profound insights into the dynamical behavior of the interconnected system featuring a neural networks (NNs) controller in juxtaposition with the closed-loop reference model. As delineated in Definition \ref{difer}, SSE serves as a robust estimator for the maximal deviation between the outputs of the two systems. Identifying the bounds of SSE enhances the confidence level for users considering the adoption of these control schemes.
\begin{theo}[Worst-case SSE]\label{theo2}
Let \(\Delta_{\delta}\) and \(\Delta_\epsilon\) satisfy the Integral Quadratic Constraint \(IQC(\Xi, M)\) and assume that the extended system \(\Sigma_{ex}\) as defined in (\ref{extendedsystem}) is well-posed. Suppose there exists a scalar \(\lambda > 0\) and a continuously differentiable matrix \(P = P^T\) with appropriate dimensions such that \(P \geq 0\), and for all \((s, \dot{s}) \in \Omega \times \dot{\mathcal{S}}\) (omitting \(s\) and \(\dot{s}\) for brevity):
\begin{equation}\label{hh}
\begin{bmatrix}
P\mathcal{A} + \mathcal{A}^T P + \partial P & P\mathcal{B}_1 & P\mathcal{B}_2 \\
\mathcal{B}_1^T P & 0 & 0 \\
\mathcal{B}_2^T P & 0 & -I
\end{bmatrix} + \lambda
\begin{bmatrix} \mathcal{C}_1^T \\ \mathcal{D}_{11}^T \\ \mathcal{D}_{12}^T \end{bmatrix}
M
\begin{bmatrix} \mathcal{C}_1 & \mathcal{D}_{11} & \mathcal{D}_{12} \end{bmatrix} < 0,
\end{equation}
and
\begin{equation}\label{ach}
\begin{bmatrix} P & \mathcal{C}_2^T \\ \mathcal{C}_2 & \sigma^2 I \end{bmatrix} > 0.
\end{equation}
Then, the Supreme Square Error (SSE) satisfies:
\begin{equation}
\sup_{s \in \Omega, 0 \neq \eta \in \mathcal{L}_2, \zeta(0) = 0} \left[ \frac{||y - \hat{y}||_\infty}{\sqrt{||d||_2^2 + ||\hat{x}||_2^2}} \right] < \sigma.
\end{equation}
\end{theo}
\begin{proof}
Considering the signals \((\chi, z, \eta, q, r)\) generated by the extended system as outlined in \ref{extendedsystem}, where the input \(\eta\) belongs to \(\mathcal{L}_2([0, \infty))\) and the parameter trajectory \(s\) resides within \(\Omega\), all under the assumption of zero initial conditions. The well-posedness of the extended Linear Parameter Varying (LPV) system ensures that all signals \((\chi, z, \eta, q, r)\) are well-defined. Furthermore, it is assumed that the uncertainties \(\Delta_{\delta}\) and \(\epsilon\) adhere to the Integral Quadratic Constraints (IQCs) specified by \((\Xi, M)\), necessitating that the signal \(r\) satisfies the time-domain expression of the IQCs in (\ref{iqctime2}) for any \(T > 0\).

A storage function \(V : \mathbb{R}^{n_{\chi}} \times \Omega \rightarrow \mathbb{R}_+\) is defined by \(V(\chi, s) = \chi^T P(s) \chi\). Additionally, the state-space representation of the virtual filter described in (\ref{iqcfilter}) leads to (\ref{str}). 

Evaluating the perturbed Linear Matrix Inequality (LMI) in (\ref{hh}) at \((s(t), \dot{s}(t))\) and performing multiplication from the left and right by \([\chi^T, q^T, \eta^T]\) and its transpose \([\chi^T, q^T, \eta^T]^T\), respectively, results in the following dissipation inequality:
\begin{equation}\label{eq88}
    \nabla_{x_{\chi}}V \dot{\chi}+ \nabla_{x_{s}}V \dot{s}\leq (1-\varepsilon) \eta \eta^T-\lambda rMr^T
\end{equation}
This implies that the dissipation inequality presented in (\ref{eq88}) is directly equivalent to the perturbed linear matrix inequality (LMI) delineated in (\ref{hh}). Proceeding with this understanding, integrating the dissipation inequality obtained in (\ref{eq88}) along the state/parameter trajectory from the initial time \(t = 0\) to an arbitrary upper bound \(t = T\) and taking into account that the initial condition for \(\chi\) is \(\chi(0) = 0\), we arrive at the following conclusion:
\begin{equation}\begin{array}{ll}
    V (\chi(T),s(T))\leq(1-\varepsilon)\int\limits_{0}^{T}\ \ \eta(\tau)\eta(\tau)^T\,\mathrm{\textit{d} \tau}-\lambda\int\limits_{0}^{T}\ \ r(\tau)Mr(\tau)^T\,\mathrm{\textit{d} \tau}
    \end{array}
\end{equation}
It follows from the IQCs condition (\ref{iqctime2}), that
\begin{equation}\label{eq90}
    V (\chi(T),s(T))\leq (1-\varepsilon)\int\limits_{0}^{T}\ \ \eta(\tau)\eta(\tau)^T\,\mathrm{\textit{d} \tau}
\end{equation}
Now, applying Schur complement on (\ref{ach}) and  multiplying left and right by $\chi, ~z$ and its transpose $\chi^T$, respectively, results:
\begin{equation}\label{eq91}
    \dfrac{1}{\sigma^2}z(t)z^T(t)\leq \chi^T(t)P(s)\chi(t)~~\forall t>0
\end{equation}
Evaluating (\ref{eq91}) at $t = T$, and applying (\ref{eq90})  yields:
\begin{equation}
    \dfrac{1}{\sigma^2}z(T)z^T(T)\leq  (1-\varepsilon)\int\limits_{0}^{T}\ \ \eta(\tau)\eta(\tau)^T\,\mathrm{\textit{d} \tau}
\end{equation}
Knowing the $\varepsilon>0$,  therefore
\begin{equation}
    \dfrac{1}{\sigma^2}z(T)z^T(T)\leq \int\limits_{0}^{T}\ \ \eta(\tau)\eta(\tau)^T\,\mathrm{\textit{d} \tau}
\end{equation}
 Considering the supremum over $T$ to demonstrate  that $||z||_{\infty}\leq \sigma ||\eta||_2$. Knowing that this result holds for any given input $\eta\in \mathcal{L}_2([0, \infty))$,
admissible trajectory solution $s\in \Omega$, and uncertainty $\Delta_{\delta}, \epsilon \in IQC(\Xi, M)$. Thus, 
 \begin{equation}
\sup_{s \in \Omega.}\sup_{0\neq \eta\in \mathcal{L}_2, \zeta(0)=0}  \dfrac{\left|\left|z(t)\right|\right|_\infty}{\left|\left|\eta(t)\right|\right|_2}<\sigma
 \end{equation}
Having $   \left|\left|z(t)\right|\right|_\infty=||y(t)-\hat{y}(t)||_\infty $ and $ \left|\left|\eta(t)\right|\right|_2=  \sqrt{||d(t)||_2^2  + ||\hat{x}(t)||_2^2}$, this implies
\begin{equation}
  \dfrac{\left|\left|z(t)\right|\right|_\infty}{\left|\left|\eta(t)\right|\right|_2}=  \dfrac{||y-\hat{y}||_\infty}{\sqrt{||d||_2^2  + ||\hat{x}||_2^2}}
\end{equation}
\end{proof}
Theorem \ref{theo2} holds significant importance as it establishes the maximum expected deviation of the closed-loop system, controlled by a neural network (NN) controller, from the closed-loop reference model. This crucial insight provides a rigorous foundation for evaluating the robustness and reliability of NN controllers, instilling confidence in their deployment by demonstrating their capability to minimise dynamical error and ensure optimal system performance. 

Moreover, the worst-case Supreme Square Error (SSE) serves as a key metric for assessing the potential emergence of undesirable behaviours within the closed-loop system under the influence of the designed NN controller. By quantifying this upper bound, Theorem \ref{theo2} plays a fundamental role in forecasting system reliability and performance, offering a predictive framework for the informed and secure implementation of NN-based control strategies in uncertain environments.

\section{ Validation and Numerical Simulation }\label{sec4}
In order to validate the theoretical results we have proved in the previous section, we will conduct tests using two different problems. Our first problem involves controlling a robot arm with only a single output, called the Single-Link Robot Arm Control Problem. A simpler and more focused scenario will allow us to evaluate the effectiveness of our approach. Further, we will also test our approach on the more complex problem of Deep Guidance and Control of Apollo Lander, which involves multiple outputs. During the descent and landing phase of the Apollo Lander, precise control and guidance is required. We can further assess the robustness and adaptability of our approach by examining this more complex scenario with multiple outputs.

Moreover, the LPV system (\ref{extendedsystem}) where its  matrices of state space depend on a time-varying parameter vector $s(t): \mathbb{R}_+\rightarrow \mathcal{S} $ and $\mathcal{S} \subset \mathbb{R}^n$. $s(t)=[d(t)^T~c(t)^T]$ a continuously differentiable function of time where the parameter rates of variation where Let \(c\) the trajectory lies on the curve \(\vartheta (\nu) = x + \nu(\hat{x}- x)\) for some \(\nu \in [0, 1]\). $\dot{s}(t): \mathbb{R}_+\rightarrow \dot{\mathcal{S}}$. $\dot{\mathcal{S}}$ is a hyperrectangle given by: 
\begin{equation}\label{dots}
    \dot{\mathcal{S}}=\left\{ \rho \in \mathbb{R}^n~| ~
    \underline{\rho}_i \leq \rho_i \leq \overline{\rho}_i,~i=1\cdots n\right\}
\end{equation}
because of the property of the activation function  of the controller in quantity $\Lambda(s)$ allows us to express the set of a bounded convex domain
\(\mathcal{G}_{m}^{n}\) of the matrix $\Lambda(s)$ as follows:
\(\mathcal{G}_{m}^{n} = \left\{ a_{\text{ij}} \leq \frac{\partial\pi_{i}}{\partial y_{j}} \leq b_{\text{ij}},\ i = 1,\cdots,n,\ j = 1,\cdots,m \right\}\)
and the set of vertices is defined by
\begin{equation} \label{ver}
    \mathcal{V}_{\mathcal{G}_{m}^{n}} = \left\{ \mathcal{V =}\left\lbrack \mathcal{V}_{11},\cdots,\mathcal{V}_{1m},\cdots,\mathcal{V}_{\text{nm}} \right\rbrack:\mathcal{V}_{\text{ij}} \in \left\{ a_{\text{ij}},\ b_{\text{ij}} \right\} \right\}
\end{equation}

In this study, we utilise the IQClab Toolbox\footnote{\nolinkurl{www.iqclab.eu}}, a MATLAB-based extension of the Robust Control Toolbox , designed for robustness analysis and control design in uncertain and linear parameter-varying (LPV) systems. As described by \cite{iqclab} and compared with another toolbox such as  IQC-$\beta$ \cite{iqcb}, IQClab provides a versatile framework incorporating integral quadratic constraint (IQC) methodologies, enabling efficient model reduction, control switching schemes, and performance weighting function generation. The modular architecture of the toolbox facilitates seamless integration with various solvers and parsers, making it a powerful tool for developing and implementing new linear matrix inequality (LMI) and IQC-based algorithms. Given its adaptability and ease of use, IQClab plays a crucial role in our analysis, offering a robust computational environment for evaluating system stability and performance.
\subsection{\color{blue} Single-Link Robot Arm Control Problem}
\subsubsection{\color{blue} Step 1: Problem Setting} The simulation validation that we conduct in this section uses the worst-case of RISE metric resulting from Theorem 1 since the open source tool we adopt for simulation calculations can only provide $\mathcal{L}_2$ to $\mathcal{L}_2$ gain.

Consider the nonlinear single-link robot arm example, shown in Fig. \ref{coco}, with parameters $ m = 0.15 kg$, length $l = 0.5 m,$ and friction coefficient
$\mu = 0.5 Nms/rad$. The state space representation  of the arm's motion is given by:
\begin{equation}\label{arm}
    \left\{\begin{array}{ll}\dot{\omega}(t)&= -10\sin(\theta(t))-2\omega(t)+\tau(t)\\
    \dot{\theta}(t)&=\omega(t)
\end{array}\right.
\end{equation}
\begin{figure}[thpb]
      \centering
   \includegraphics[scale=0.079]{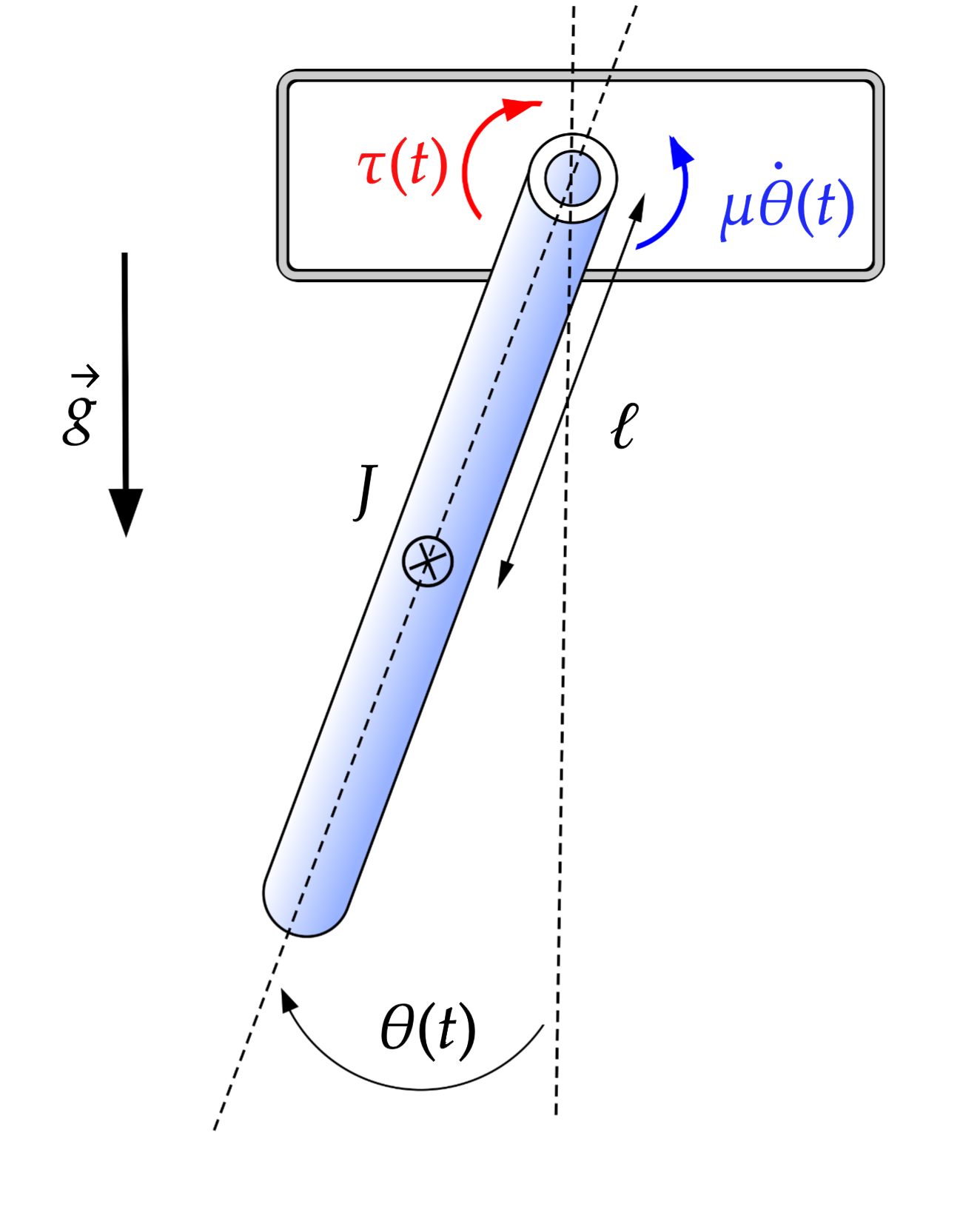}
\caption{ Single-link robot arm. }
     \label{coco}
   \end{figure}
 \subsubsection{\color{blue} Step 2: Linear Reference Model}
We begin by training a neural controller with the objective of ensuring that the arm, when governed by this controller, dynamically replicates the behavior of the closed-loop reference model defined as:
\begin{equation}\label{reference}
    \left\{\begin{array}{ll}\dot{\omega}_r(t)&= -9\theta_r(t)-6\omega_r(t)+9 r(t)\\
    \dot{\theta}_r(t)&=\omega_r(t) 
\end{array}\right.
\end{equation}
According to (\ref{arm}), achieving the desired behaviour requires an ideal classical controller that ensures the system's dynamic response aligns with the prescribed closed-loop reference model.
\begin{equation}\label{contr1}
\pi^*\left(r(t),  \theta_r(t)\right)  =   10 \theta_r(t)+4\omega_r(t)-9r(t)-10 \sin(\theta_r(t))
\end{equation}
According to Equation (\ref{arm}), achieving the desired behaviour requires an ideal classical controller that ensures the system's dynamic response aligns with the prescribed closed-loop reference model:

\begin{equation}\label{contr1}  
\pi^*\left(r(t),  \theta_r(t)\right)  =   10 \theta_r(t) + 4\omega_r(t) - 9r(t) - 10 \sin(\theta_r(t))  
\end{equation}  

where \(\tau_r(t) = \pi^*\left(r(t),  \theta_r(t)\right) \). This implies that the dynamic behaviour of the system (\ref{arm}) in the closed-loop configuration, following the training of the neural controller, approximates the desired closed-loop dynamics (\ref{reference}). The ideal scenario would be for the system governed by the neural network controller (\ref{arm}) to exactly match the closed-loop reference model (\ref{reference}) at all times \(t\). However, achieving this perfect alignment is practically infeasible due to inherent system nonlinearities and modelling limitations.

The objective of this example is to determine the worst-case relative squared difference, \(\gamma\), between the dynamic response of the robot arm under neural network control, accounting for its nonlinear characteristics, and the ideal closed-loop reference model.
\begin{equation}
    \sup_{s \in \Omega.}\sup_{0 \neq [\tau(t) ~\theta_r(t)]^T \in \mathcal{L}_{2},\zeta(0) = 0} \frac{\left| \left|\theta\left( t \right) - \theta_r(t)\right| \right|_2}{\left| \left|[r(t) ~\theta_r(t)]^T\right| \right|_2}  < \gamma
\end{equation}
To this end, let assume the static nonlinearity $ q(t)=\Delta_{\delta} (\theta(t)) = \theta(t)- \sin(\theta(t))$ which is slope-restricted and sector bounded. We rearrange   (\ref{arm}) then into the form: 
\begin{equation}\label{ara}
    \left\{\begin{array}{ll}\dot{\omega}(t)&= -9\theta(t)-2\omega(t)+\tau(t)+10q(t)\\
    \dot{\theta}(t)&=\omega(t),\\
    ~p(t)&=\theta(t),\\
    \Delta_{\delta} (\theta(t)) &=q(t)= \theta(t)a- \sin(\theta(t))
\end{array}\right.
\end{equation}
\subsubsection{\color{blue} Step 3: Error Dynamical System}
Now, we generate the dynamical error system using the closed-loop reference model (\ref{reference}) and the new mathematical form of the motion of the arm (\ref{ara}). Let $\delta \omega=\omega-\omega_r$,  $\delta\theta=\theta-\theta_r$ and $\delta \tau=\tau-\tau_r$  then
\begin{eqnarray}
\dfrac{d}{dt}\begin{bmatrix} \delta\omega\\ \delta\theta\end{bmatrix}=\begin{bmatrix} -6 & -9\\ 1 & 0\end{bmatrix}\begin{bmatrix} \delta\omega\\ \delta\theta\end{bmatrix}+\begin{bmatrix} 10\\ 0\end{bmatrix}q(t)+\begin{bmatrix} 1\\ 0\end{bmatrix}\delta\tau(t),~~\zeta(t)=\begin{bmatrix} \delta\omega\\ \delta\theta\end{bmatrix}
\end{eqnarray}
The neural controller architecture adopted in this example is given in Fig. \ref{nnc}. Thus, the output of the neural controller is given by :
\begin{equation}\label{contr2}
\pi\left(r(t), \theta(t)\right) = w_{4}\tanh\left(w_{1}\theta\left( t \right) + w_{2}r\left( t \right) + w_{3} \pi\left(r(t), \theta(t)\right) +  b_{1} \right) + b_{2}\end{equation}
where $\tau(t)=\pi\left(r(t), \theta(t)\right) $.
\begin{figure}[thpb]
      \centering
   \includegraphics[scale=0.15]{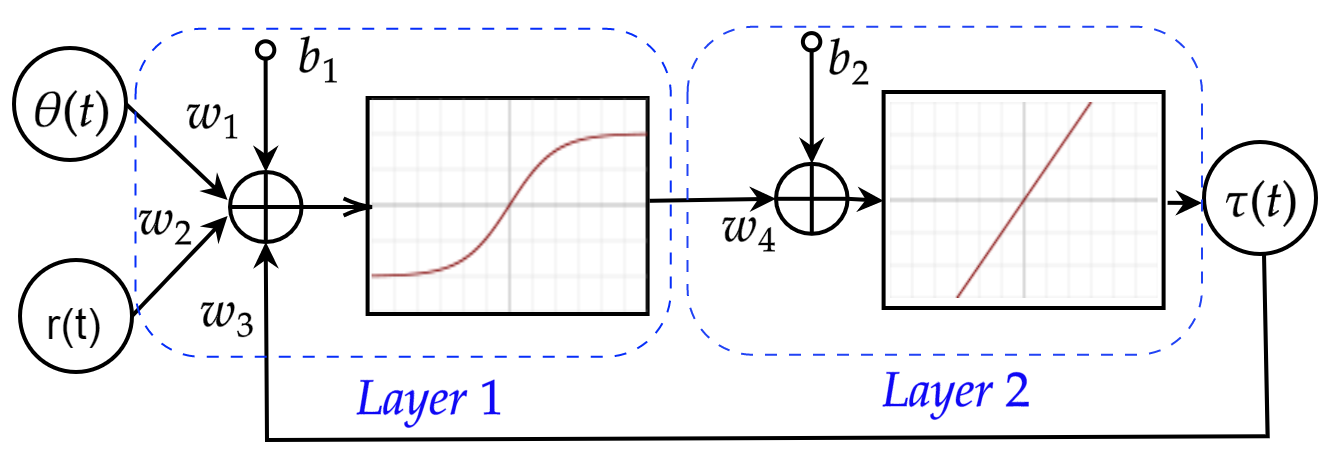}
\caption{ Neural controller $\pi^{*}$ architecture.  }
      \label{nnc}
   \end{figure}
Using, Differential Mean Value Theorem \cite{ref23} introduced in lemma \ref{theo1}, can convert the dynamical error system to an LPV system without nonlinearities of the neural controller. According to lemma \ref{lemma4}, the difference between the two controller $\Delta_{\delta}\tau(t)=\mu(t)$ given in  (\ref{erroriqc}) can be expressed as : 
 \begin{equation} \label{eq51}
\delta\tau(t)=\mu(t)=\Lambda(s) \zeta(t)+\frac{\partial \alpha}{\partial \theta_r}(0,0)\theta_r(t)+\frac{\partial \alpha}{\partial d }(0,0) d(t)+\epsilon (t),
 \end{equation}
where the trajecory $s(t)=[d(t), c(t)]^T$, the LPV matrix $ \Lambda(s(t))=\nabla_x \pi(d(t),c(t))$ is the Jacobian matrix 
 of the neural network, $\epsilon (t)=\mathcal{O}(d(t),\hat{x}(t))$ and  the training error  $\alpha(d(t),\hat{x}(t))=\pi(d(t), \hat{x}(t))-\pi^*(d(t), \hat{x}(t))$.
In this case, we have $ C=\begin{bmatrix} 0& 1\end{bmatrix}$, $\zeta(t) =\begin{bmatrix} \delta\omega(t)& \delta\theta(t)\end{bmatrix}^T$,  $\hat{y}(t)=\theta_r$ and $y(t)=\theta(t)$. Therefore,
\begin{equation} \label{eq53}
\delta\tau(t)=\Lambda(s)\delta\theta(t) +\frac{\partial \alpha}{\partial \theta_r}(0,0)\theta_r(t)+ \frac{\partial \alpha}{\partial r }(0,0) r(t)+\epsilon (t)
 \end{equation}
where $\pi^*(.,.)$ and $\pi(.,.)$ are given in (\ref{contr1}) and (\ref{contr2}), respectively. and 
 \begin{equation}
 \alpha(r(t),\theta_r(t))=\pi(r(t), \theta_r(t))-\pi^*(r(t), \theta_r(t))
\end{equation}
Now,  to calculate $\Lambda(s)$ let consider the function of he neural controller (\ref{contr2})
\begin{equation}
\pi(r(t), \theta(t)) = w_{4} \tanh\left(w_{1} \theta(t) + w_{2} r(t) + w_{3} \pi(r(t), \theta(t)) + b_{1} \right) + b_{2}.
\end{equation}
Differentiating both sides with respect to \( \theta \) and applying implicit differentiation, we obtain:
\begin{equation}
 \frac{\partial \pi}{\partial \theta} (r(t), \theta(t))= w_4 \text{sech}^2(\varrho(r(t), \theta(t))) \left(w_1 + w_3 \frac{\partial \pi}{\partial \theta}(r(t), \theta(t)) \right),
\end{equation}
where
\begin{equation}
 \varrho(r(t), \theta(t)) = w_{1}\theta(t) + w_{2}r(t) + w_{3} \pi(r(t), \theta(t)) + b_{1}.
\end{equation}
Rearranging and solving for \( \frac{\partial \pi}{\partial \theta} \), we obtain:
\begin{equation}\label{lamda}
\Lambda(s(t))=\frac{\partial \pi}{\partial \theta} (r(t), c(t))= \frac{w_4 w_1 \text{sech}^2(\varrho(r(t), c(t)))}{1 - w_4 w_3 \text{sech}^2(\varrho(r(t), c(t)))}.
\end{equation}
where the trajectory \(c \in \mathbb{Co}(\theta,\theta_r)\) lies on the curve \(\vartheta (\nu) = \theta + \nu(\theta_r- \theta)\) for some \(\nu \in [0, 1]\). Also, it is clear that  
\begin{eqnarray}\label{alpha}
    \frac{\partial \alpha}{\partial \theta_r}(0,0)=\frac{w_4 w_1(1-\tanh^2(b_1))}{1-w_4w_3(1-\tanh^2(b_1))}, ~\frac{\partial \alpha}{\partial r }(0,0)=\frac{w_4 w_2(1-\tanh^2(b_1))}{1-w_4w_3(1-\tanh^2(b_1))}+9
\end{eqnarray}

Therefore, the dynamical error system became:
\begin{eqnarray}\left\{ \begin{array}{ll}
 \dfrac{d}{dt}\begin{bmatrix} \delta\omega\\ \delta\theta\end{bmatrix}=\begin{bmatrix} -6 & -9\\ 1 & 0\end{bmatrix}\begin{bmatrix} \delta\omega\\ \delta\theta\end{bmatrix}+\begin{bmatrix} 10\\ 0\end{bmatrix}q(t)+\begin{bmatrix} 1\\ 0\end{bmatrix}\delta\tau(t),\\\\
\delta\tau(t)=\Lambda(s)\delta\theta(t) +\frac{\partial \alpha}{\partial \theta_r}(0,0)\theta_r(t)+ \frac{\partial \alpha}{\partial r }(0,0) r(t)+\epsilon (t),
\end{array}\right .
\label{armerror}
\end{eqnarray}
where the parameters $\Lambda(s)$, $\frac{\partial \alpha}{\partial \theta_r}(0,0)$ and $\frac{\partial \alpha}{\partial \theta_r}(0,0)$ are given in (\ref{lamda}) and (\ref{alpha}).
\subsubsection{\color{blue} Step 4: Characterisation of the nonlinearities $\Delta_{\delta}$ and $\Delta_\epsilon$ using IQCs}
To choose the type of IQCs of the uncertainties considered for this simulation example, let us start with the static nonlinearity $ \Delta_{\delta} (\theta) = \theta- \sin(\theta)$. It is clear  from  the drawn figure Fig. \ref{uncer1} that the uncertainty $\Delta_{\delta} (\theta)$ is slope-restricted and sector bounded. Fig. \ref{uncer1} provides the IQCs characterisation of this nonlinearity and plots the graph  $\Delta_{\delta}=f (\theta)$ where $\Delta_{\delta}(\theta(t))=\varphi(\theta,t)$ and $\tilde{\alpha} \theta(t) \leq\varphi(\theta,t)\leq \tilde{\beta} \theta(t)$ for all $\theta\in[-\dfrac{\pi}{2},~ \dfrac{\pi}{2}]$.  Simple calculation of the slope gives $\tilde{\alpha}=0$ and $\tilde{\beta}=0.364$ representing the bounds of the uncertainty $\Delta_{\delta}$. 
\begin{figure}[thpb]
      \centering
   \includegraphics[scale=0.2]{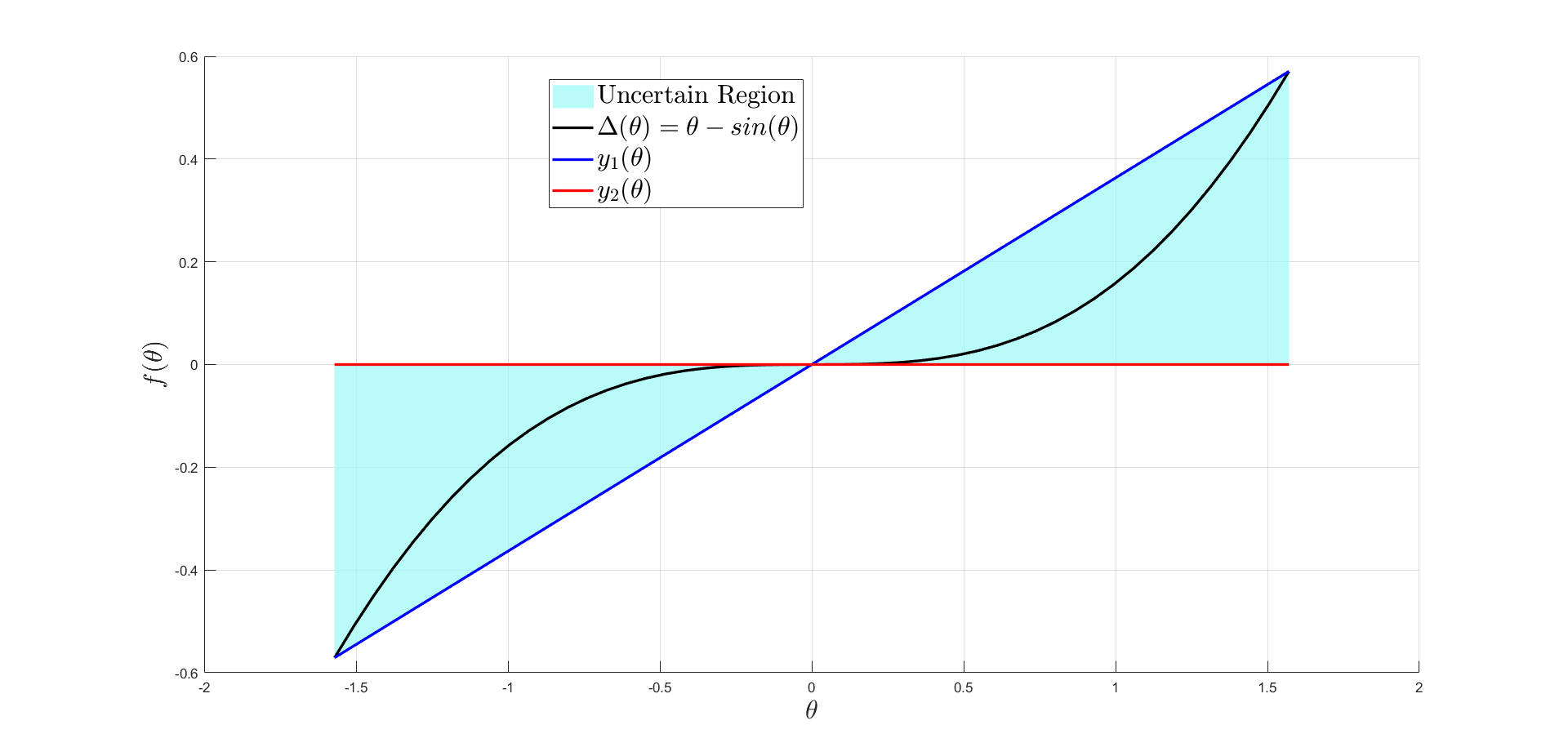}
\caption{ IQCs characterisation  of the nonlinearity $\Delta_{\delta}$.}\label{uncer1}
   \end{figure}
 Similarly, we have the nonlinearity $\epsilon(t)=\Delta_\epsilon (r(t),\theta_r(t))=\mathcal{O}(r(t),\theta_r(t))$ which is reminder of the first order approximation of  $\alpha(r(t),\theta_r(t))=\pi(r(t), \theta_r(t))-\pi^*(r(t), \theta_r(t))$, where $\pi^*$ and $\pi$ are given in (\ref{contr1}) and (\ref{contr2}), respectively. According to (\ref{tho}) we have 
 \begin{eqnarray}
        |\epsilon(t)| = |\mathcal{O}(d(t), \hat{x}(t))| \leq \kappa_1 |r(t)| + \kappa_2 |\theta_r(t)|.
 \end{eqnarray}
for all $\theta_r,~r\in[-\dfrac{\pi}{2},~ \dfrac{\pi}{2}]$. 
\subsubsection{\color{blue} Step 5: IQCs based analysis results}
Now, we would like to compute an upper bound on the $\mathcal{L}_2$-gain from  the input vector $[r, ~\theta_r]^T$ to the output error $\theta-\theta_r$. We use IQClab toolbox to express and formulate the $\mathcal{L}_2$ gain estimation problem.  We only need the two linear systems, the closed-loop reference model and the dynamical error system, since the output of the closed-loop system with the NN controller is the sum of the outputs of those two systems. The bounds of the uncertainty $\Delta_{\delta}$ and the nonlinearity $\Delta_\epsilon$ derived in the previous sub-section are also provided as inputs to the IQClab Toolbox.

According to the results given by IQClab for the model in (\ref{armerror}) corresponding to the original Single-Link Robot Arm Control Problem:
\begin{equation}\label{gamma}
    \sup_{s \in \Omega.}\sup_{0 \neq [r(t) ~\theta_r(t)]^T \in \mathcal{L}_{2},\zeta(0) = 0} \frac{\left| \left|\theta\left( t \right) - \theta_r(t)\right| \right|_2}{\left| \left|[r(t) ~\theta_r(t)]^T\right| \right|_2}  < \gamma = 0.04009
\end{equation}

According to \cite{oxford}, the time Integral of the Square of the Error (ISE) is commonly used in control systems as a measure of system performance. Thus, \(\gamma\) in (\ref{gamma}) represents the worst-case mean integrated squared error between the closed-loop system with the trained neural controller and the closed-loop reference system. In other word, the value of \(\gamma\) quantifies the worst-case relative error between the closed-loop reference model, which serves as the basis for training the neural network controller, and the uncertain system operating with the designed neural controller. To illustrate further the results of our controlled Single-Link Robot Arm with our designed neural controller, we propose two study cases: A and B. The reference of the first case study is periodically applied to the uncertain system with a neural controller and the reference closed-loop model.   Using the equation given in  (\ref{gamma}), we can plot the expected region by utilise the information about the worst-case of the error and the known linear reference model as follow: 
\begin{equation}\label{gamma1}
       \left\|\theta\left( t \right) - \theta_r(t)\right\|_2<\ \gamma \left\| [r(t)~ \theta_r(t)]^T\right\|_2=0.04009 \sqrt{\left\| r(t)\right\|_2^2+\left\| \theta_r(t)\right\|_2^2}
 \end{equation}
Knowing that, 
\begin{eqnarray}
\theta_r(t)- \left\|\theta\left( t \right) - \theta_r(t)\right\|_2 \leq\theta\left( t \right) \leq \theta_r(t)+ \left\|\theta\left( t \right) - \theta_r(t)\right\|_2
\end{eqnarray}
Using (\ref{gamma1}),  we conclude that 
\begin{eqnarray}
\theta_r(t)- 0.04009 \sqrt{\left\| r(t)\right\|_2^2+\left\| \theta_r(t)\right\|_2^2}\leq\theta\left( t \right) \leq \theta_r(t)+ 0.04009 \sqrt{\left\| r(t)\right\|_2^2+\left\| \theta_r(t)\right\|_2^2}
\end{eqnarray}
Using this equation is of significant importance, as it implies that, by relying solely on the known quantities \(r(t)\) and \(\theta_r(t)\), we can predict the behaviour of the uncertain system controlled by the neural network. More specifically, it allows us to determine the expected regions and boundaries within which the system’s output will remain, despite uncertainties and nonlinearities.

This insight is particularly crucial in control applications, as it enables the assessment of the robustness and reliability of the neural network controller to maintain stability and performance under varying conditions. By characterising these boundaries, we gain a deeper understanding of the extent to which the trained neural controller approximates the desired closed-loop reference model. Moreover, this formulation provides a foundation for further optimisation, enabling the refinement of the neural controller to enhance its adaptability and accuracy in real-world applications. This worst-case result can be used to improve the robustness of the NN controller by making the dynamical error smaller, which in turn will guarantee better performance of the NN controller when implemented. This worst-case measure is also an indication of the appearance likelihood of any undesirable behaviour in the closed-loop system when adopting the designed NN controller.

\begin{figure}[thpb]
      \centering
   \includegraphics[scale=0.4]{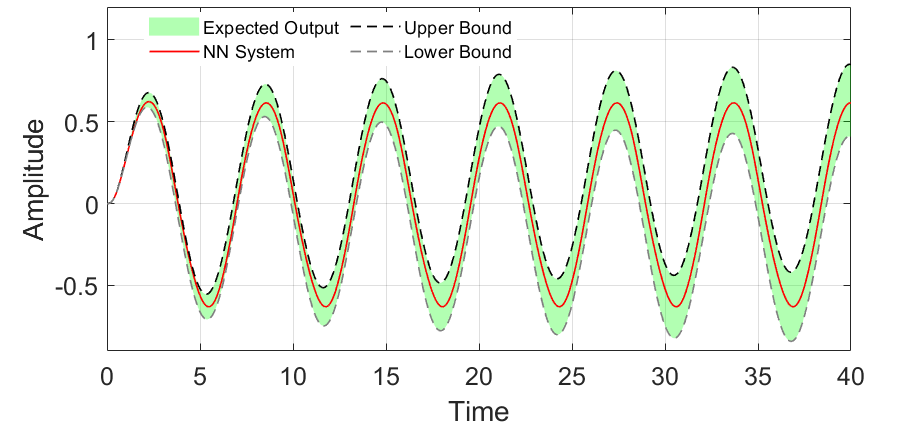}
      \caption{\textbf{Case Study A}:  Plot for the closed-loop system with the NN within the region of expectation .}
     \label{caseb1}
   \end{figure}
   \begin{figure}[thpb]
      \centering
   \includegraphics[scale=0.4]{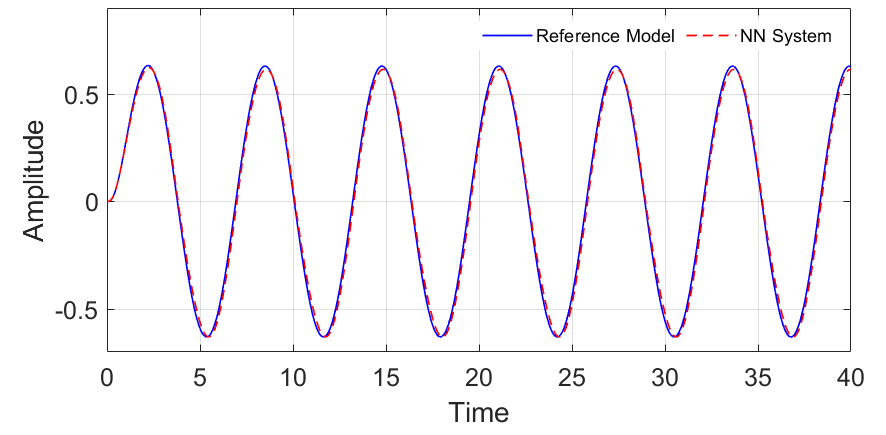}
\caption{\textbf{Case Study A}:  The controllers $\pi$ and $\pi^*$  outputs.}
     \label{caseb2}
   \end{figure}
 \begin{figure}[thpb]
      \centering
   \includegraphics[scale=0.4]{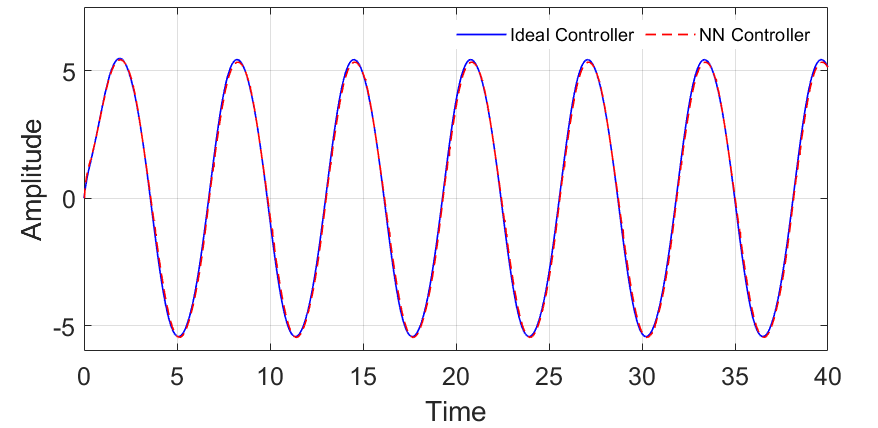}
      \caption{\textbf{Case Study A}:  Plot of the output of reference model response $\theta_r$ indicated in blue and the one of the closed-loop uncertain robotic arm system with neural controller $\theta$ indicated in red.}
     \label{caseb3}
   \end{figure}
\begin{figure}[thpb]
      \centering
   \includegraphics[scale=0.4]{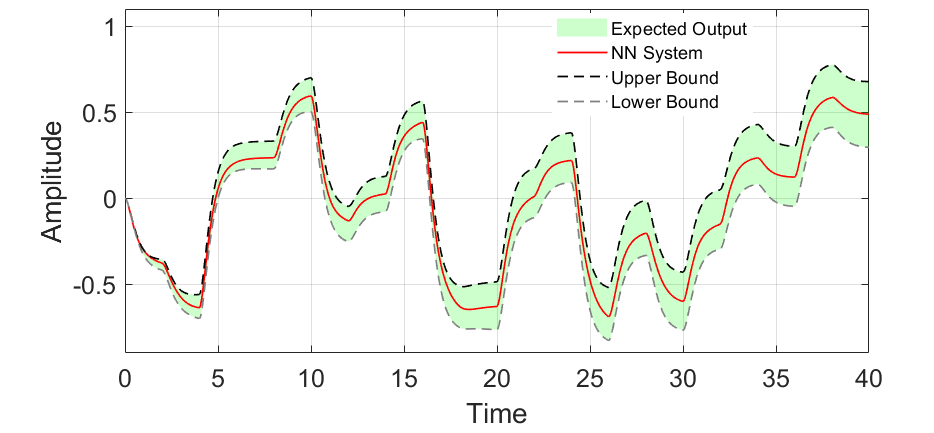}
      \caption{\textbf{Case Study B}:  Plot for the closed-loop system with the NN within the region of expectation .}
     \label{casea1}
   \end{figure}
   \begin{figure}[thpb]
      \centering
   \includegraphics[scale=0.4]{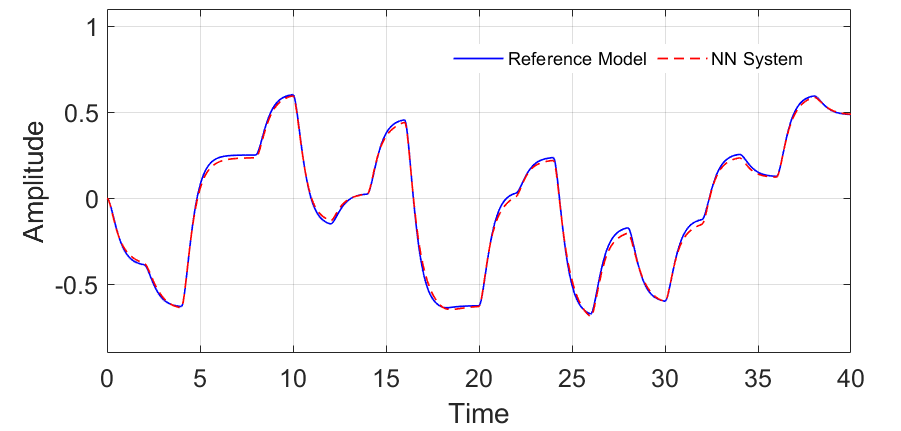}
\caption{\textbf{Case Study B}:  The controllers $\pi$ and $\pi^*$  outputs.}
     \label{casea2}
   \end{figure}
 \begin{figure}[thpb]
      \centering
   \includegraphics[scale=0.4]{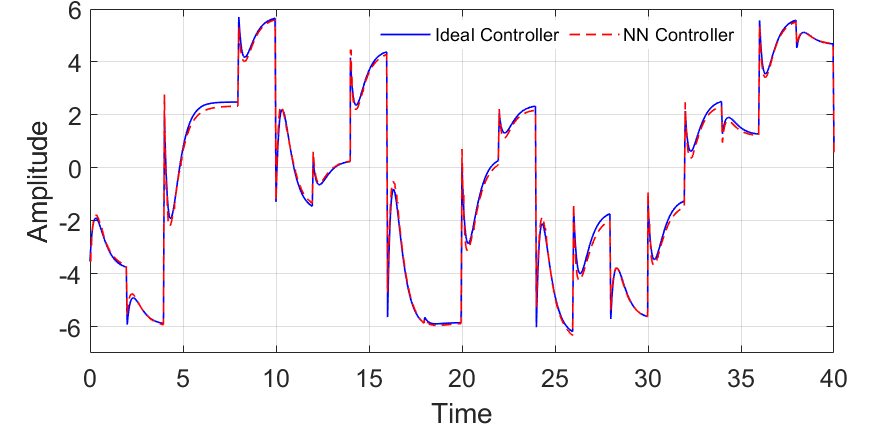}
      \caption{\textbf{Case Study B}:  Plot of the output of reference model response $\theta_r$ indicated in blue and the one of the closed-loop uncertain robotic arm system with neural controller $\theta$ indicated in red.}
     \label{casea3}
   \end{figure}
Figure \ref{caseb1} illustrates the response of the system with the trained neural network (NN) controller, where it is evident that the output remains within the expected bounded region. In Figure \ref{caseb2}, the responses of both the closed-loop reference model and the uncertain plant with the NN controller are presented for this first case study. From this figure, it can be observed that the uncertain plant with the NN controller \textbf{closely follows} the response of the closed-loop reference model.  

Additionally, the output of the neural controller remains highly similar to that of the reference controller used in the closed-loop reference model, as depicted in Figure \ref{caseb3}. The observations made for Figures \ref{caseb1}–\ref{caseb3} are equally applicable to the non-periodic case study B, illustrated in Figures \ref{casea1}–\ref{casea3}, where the system’s response with the NN controller consistently remains within the expected performance region. These results further validate our findings, demonstrating the effectiveness of the neural network controller in maintaining stability and ensuring system performance.

\subsection{\color{blue}Deep Guidance and Control of Apollo Lander}
 In the classical Guidance and Control of Apollo Lander described in  \cite{iros21} ,   GC system generates trajectory and thrust commands to manoeuvre the lander for a given initial position, velocity, and desired landing target.  The control part is compensated for by the Martian atmosphere and converted into inertial coordinates before the position commander uses the information \cite{iros20}.  
\subsubsection{\color{blue} Step 1: Problem Setting}
 We perform the test of the trained network in closed-loop using the parameters listed in Table \ref{tab1} with the aim of validating its performance.  In terms of applying aerodynamic forces and torques which are external disturbances  on the lander, coefficient of drag $C_D$ is $2.0$, and it is assumed to be constant, whereas surface-relative head wind is in the range of $20 m/s$. Furthermore, the air density $\rho$ is assumed to be $0.023 kg/m^3$ throughout. Under this scenario, the worst-case aerodynamic resistance during the descent phase is less than 2\% of the vehicle's maximum thrust. Due to the close-loop control action, aerodynamic forces and torques are automatically compensated.
\begin{table}[!h]
\centering
  \caption{Simulation Parameters }
\begin{tabular}{l|c}
\hline
	Parameter &     value  \\\hline
	Initial Position  $p_0$ $[m]$ &   $[-5632.2, 709.25 ,6190.5]^T$ \\ 
    Initial Velocity   $v_0$ $[m/s]$ &     $[206.48, -26.006 , -103.51]^T$   \\ 
	Initial Acceleration  $a_0$ $[m/s^2]$ &   $[-2.1124,	0.24947,	-2.6404]^T$   \\ 
	Intial Mass $m_0$ $[Kg]$ &  600.3  \\ 
 The capsule reference area  \(S\) & 5.137426149499100\\ 
 Mars Gravity $g$ $[m/s^2]$& $[0, 0, 3.725258]^T$ \\
 Engine Maximum Thrust $T$ $[N]$& 3600\\
 Maximum Throttle Level $T_{max}$& 100\%\\
  Minimum Throttle Level $T_{min}$ & 30\%\\
 \hline
\end{tabular}
    \label{tab1}
\end{table}
\paragraph{Atmosphere}
As part of the Mars environment module, the horizontal wind is included in the atmospheric model. Based on lookup tables, an atmospheric model specifies the properties of the Martin atmosphere at a specific location.  The atmospheric properties are obtained directly from the Mars Climate Database (MCD) v.5.2 at the mission’s specified landing coordinates.  MCD is derived from numerical simulations of the Martian atmosphere using General Circulation Models (GCM) and validated by observations of the Martian atmosphere. For more details see \cite{iros20}. 
\paragraph{Lander Dynamics}
In closed-loop simulations, Newton's second law is used to model a point mass's robotic lander motion, where the lander motion is expressed in the LENU frame as follow: 
\begin{equation}
    \frac{d\left\lbrack r \right\rbrack_{LENU}}{d t} = \left\lbrack v \right\rbrack_{LENU},
\end{equation}
where  (\(\left\lbrack v \right\rbrack_{LENU}\))  refer to the relative velocity of the robotic lander. In in the LENU frame and the lander mass (\(m\)),  The dynamics equation is given based on the force resultant
(\(\left\lbrack F \right\rbrack_{LENU}\)) expressed as: 
\begin{equation}
    \frac{d\left\lbrack v \right\rbrack_{LENU}}{dt} = \left\lbrack a \right\rbrack_{LENU} = \frac{\left\lbrack F \right\rbrack_{LENU}}{m}.
\end{equation}
\paragraph{Aerodynamics}
In the simulations, a lookup table is used as the aerodynamics model of the Apollo capsule.  In this module, only aerodynamic drag forces are considered
following equation:
\begin{equation}
    F_{D} = \frac{1}{2}\rho v^{2}C_{D}\left( \text{Ma} \right)S, 
\end{equation}
where \(\rho\) represents the atmosphere density, \(C_{D}\) denotes the capsule drag coefficient as a function of the Mach number (\(\text{Ma}\)),
 the capsule reference area represented by  the \(S\), and \(v\) is the capsule
velocity relative to the fluid velocity expressed in the landing site
East-North-UP frame.  
\begin{figure}[thpb]
      \centering
   \includegraphics[scale=0.33]{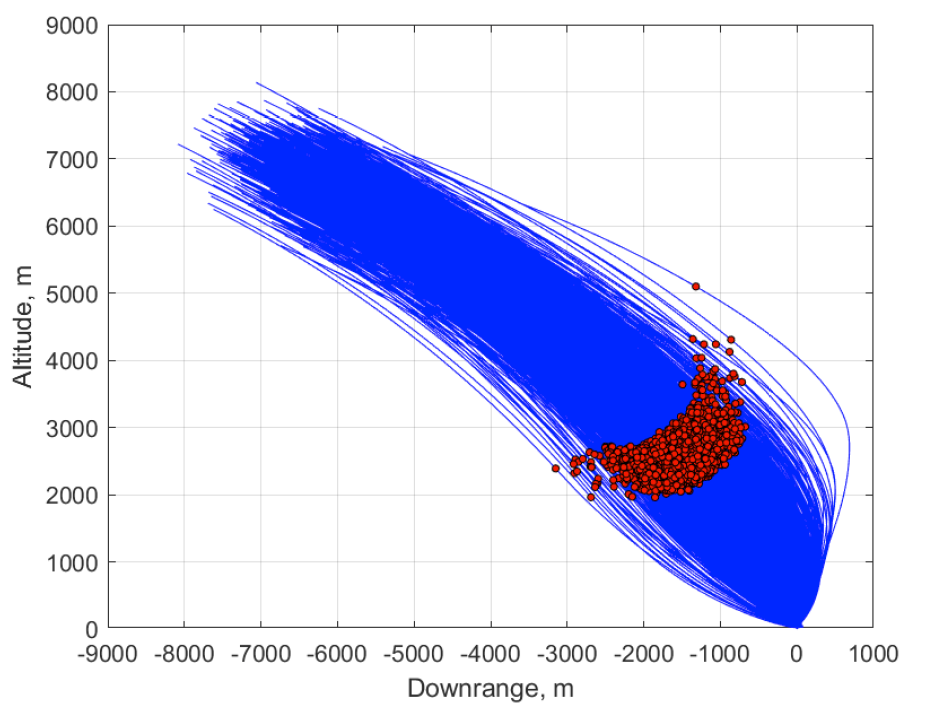}
\caption{Plot of the 1537 trajectories \cite{iros21}. The marked dots representing the trigger locations.  An optimal altitude to trigger is in the range of 2 km to 3 km, and in some trajectories, overshooting the target is the best behavior.}
      \label{data}
   \end{figure}
\paragraph{Actuators}
In the actuator model, four thrusters engines are used to generate forces concerning the center of mass of the vehicle.  The main engine responsible on the force \(F_{Thr}\) is located along with the vehicle $z$-axis, while the remaining are the sources of saturated forces  \(F_{Trq}\) are placed at a distance \(l\) from the mass center to generate pure torques along positive and negative axes. To calculate the fuel consumption, the mass flow rate is integrated:
\begin{equation}
    {\dot{m}}_{fuel} = {V_{e}}^{- 1}\left( F_{Thr} + \sum F_{Trq} \right),
\end{equation}
where   \(V_{e} = g_{0}I_{sp}\).
\paragraph{Apollo Guidance }The algorithm is based on the idea of following the reference trajectory backwards in time from the target point to the present moment. The time-to-go (\(t_{go}\)), represents the time left until the desired landing is reached, so it tends to zero as time gets closer to the goal.  In \cite{iros22}, the reference trajectory satisfies a two-boundary problem with five degrees of freedom. Accordingly, the reference trajectory must be a polynomial function of fourth or higher order.  Therefore, the reference position (\(r_{d}\) given at \(t_{go}\)  as a function of the target state as:
\begin{equation}\label{eq6}
    r_{d} = \ r_{t} + v_{t}t_{go} + a_{t}\frac{{t_{go}}^{2}}{2} + j_{t}\frac{{t_{go}}^{3}}{6} + s_{t}\frac{{t_{go}}^{4}}{24},
\end{equation}
where \(r_{t}\) refers to the target position, \(v_{t}\) represents
the target velocity and \(a_{t}\), \(j_{t}\),
\(s_{t}\) denote the target acceleration, jerk and snap,
respectively.
According to (\ref{eq6}),  the work of  \cite{iros7} describes the acceleration command (\(a_{cmd}\)) as a quadratic polynomial i.e.
\begin{equation}\label{eq8}
    a_{\text{cmd}} = \ C_{0} + C_{1}t_{go} + {{C_{2}t}_{go}}^{2}.
\end{equation}
The coefficients of the  polynomial are derived by solving a system of three
equations comprising (\ref{eq8}). The two other equations related to  the target velocity and position equations are obtained by by integrating the target acceleration equation:
\begin{equation}\label{eq9}
    \left\{ \begin{array}{l}
a_{t} = \ C_{0} + C_{1}t_{go} + {{C_{2}t}_{go}}^{2} \\
v_{t} = C_{0}t_{go} + \frac{1}{2}C_{1}{t_{go}}^{2} + \frac{1}{3}C_{2}{t_{go}}^{3} + v \\
r_{t} = vt_{go} + {\frac{1}{2}C}_{0}{t_{go}}^{2} + \frac{1}{6}C_{1}{t_{go}}^{3} + \frac{1}{12}C_{2}{t_{go}}^{4} + r \\
\end{array} \right.
\end{equation}
where \(a_{t}\) represents the reference acceleration vector,
\(r\) and \(v\) denote the vehicle current position and
velocity, respectively.  The following values for the
polynomial coefficients \(C_{0}\), \(C_{1}\)
and \(C_{2}\)) are obtained by solving the following matrix system:
\begin{equation}\label{eq10}
    \begin{bmatrix}
a_{t}\\
v_{t}\\
r_{t} 
\end{bmatrix}=    \begin{bmatrix}
1&t_{go}&t_{go}^{2}\\
t_{go}&\frac{1}{2}t_{go}^{2}&\frac{1}{3}t_{go}^{3}\\
\frac{1}{2}t_{go}^{2}&\frac{1}{6}{t_{go}}^{3} &  \frac{1}{12}{t_{go}}^{4}
\end{bmatrix} \begin{bmatrix}C_{0} \\ C_{1} \\ C_{2} \end{bmatrix} + \begin{bmatrix}0 \\v \\ v_{t_{go}} \end{bmatrix} .
\end{equation}
Therefore, 
\begin{equation}
    \left\{ \begin{array}{l}
C_{0} = \ a_{t} - 6\left( \frac{v_{t} + v}{t_{go}} \right) + 12\left( \frac{r_{t} - r}{{t_{go}}^{2}} \right) \\
C_{1} = \  - 6\left( \frac{a_{t}}{t_{go}} \right) + 6\left( \frac{{5v}_{t} + 3v}{{t_{go}}^{2}} \right) - 48\left( \frac{r_{t} - r}{{t_{go}}^{3}} \right) \\
C_{2} = \ 6\left( \frac{a_{t}}{{t_{go}}^{2}} \right) - 12\left( \frac{{2v}_{t} + v}{{t_{go}}^{3}} \right) + 36\left( \frac{r_{t} - r}{{t_{go}}^{4}} \right)
\end{array} \right.
\end{equation}
By  defining the vertical component of the acceleration profile as a linear function of \(t_{go}\). Hence, solving (\ref{eq8}) for \(t_{go}\) two solutions
are obtained: 
\begin{equation}
\begin{array}{l}
\mathbf{If} ~ \left| \left( a_{t} \right)_{z} \right| > 0 \\
t_{go} = \frac{2\left( v_{t} \right)_{z} + \left( v \right)_{z}}{\left( a_{z} \right)_{t}} + \sqrt{\left\lbrack \frac{2\left( v_{t} \right)_{z} + \left( v \right)_{z}}{\left( a_{t} \right)_{z}} \right\rbrack^{2} + \frac{6\left\lbrack \left(\mathbf{ r} \right)_{z} - \left( r_{t} \right)_{z} \right\rbrack}{\left( a_{t} \right)_{z}}}, \\
\mathbf{else}~\mathbf{if}\left( a_{t} \right)_{z} = 0
\\
t_{go} = \frac{3\left\lbrack \left( r_{t} \right)_{z} - \left( r \right)_{z} \right\rbrack}{\left( v \right)_{z} + 2\left( v_{t} \right)_{z}},  
\end{array}
\end{equation}
where $(.)_z$ denotes the vertical component of the variable $(.)$.

A robustness analysis of feedback systems Apollo lander based on NN controllers against potential uncertainties is presented here.  To achieve this, we need to follow the below pattern where we aiming at keeping  the output of the closed-loop system with a NN controller close to the output of an ideal and trusted closed-loop reference model when its input changes within a bounded set. 
  \subsubsection{\color{blue} Step 2: Linear Reference Model}
  We start by training a neural controller with the aim that the lander system with this neural controller tracks classical Apollo controller provided in \cite{iros21}. A linear reference model is what we want. To generate the linear closed-loop reference model, we use the Five Tau rule, which states that transient events die down after five Tau seconds. In other words, a transition from a (steady) state to another (steady) state lasts only for five Tau seconds. Let consider the following form of the reference model:
  \begin{equation}
      \left[\begin{matrix}\frac{d\hat{r}}{dt}\\\frac{d\hat{v}}{dt}\\\end{matrix}\right]=\left[\begin{matrix}\hat{v}(t)\\\ q\hat{v}(t)+\hat{u}(t)\\\end{matrix}\right]
  \end{equation}
 where 
$$q=\left(\begin{matrix}q_x&0&0\\0&q_y&0\\0&0&q_z\\\end{matrix}\right)$$

MATLAB is used to find the matrix $a$ using regression linear technique:
\begin{equation}
    q= \arg \min_{v}(\frac{1}{2}\frac{\rho}{m}v^2C_D\left(\mathrm{Ma}\right)S\ -q v(t))
\end{equation}
As a result, we obtain the matrix $a$ as follows: 
\begin{equation}
    a=\left(\begin{matrix}-0.0087         &0&0\\0&-0.0075         &0\\0&0&-0.0077\\\end{matrix}\right)
\end{equation}
The linear reference model is now complete, and we can move on to finding the best controller which has dynamical behaviour  too close to the nonlinear classical one. MATLAB is used to implement the poles placement technique (state feedback control) along with Five Tau rule. 

The neural controller in TABLE \ref{tab2} was trained with the aim that the lander dynamics with this neural controller tracks the closed-loop reference Apollo model. This means that the dynamic behaviour of the system in the closed-loop system after the training of the neural controller is close to the linear reference closed-loop dynamics. In an ideal world, the closed-loop Apollo model would be matched by the system with a NN controller at all times $t$. In reality, this is almost impossible. In this example, we are calculating the worst relative squared difference, $\gamma$, between the classical Apollo system provided in \cite{iros21} with neural control system closing the loop, while subject to its nonlinear functions, and the ideal closed-loop reference model. For this purpose, let assume the static nonlinearity $\Delta_{\delta}\left(v(t)\right)=\frac{1}{2}\frac{\rho}{m}v^2(t)C_D\left(\mathrm{Ma}\right)S\ -qv(t)$ which is slope-restricted, and sector bounded. The Lander dynamics are then rearranged into the following form: \begin{equation}
    \left\{\begin{matrix}\left(\begin{matrix}\frac{dr}{dt}\\\frac{dv}{dt}\\\end{matrix}\right)=\left[\begin{matrix}v\\qv\left(t\right)+u\left(t\right)+\Delta_{\delta}\left(v\left(t\right)\right)\\\end{matrix}\right]\\\ \Delta_{\delta}\left(v(t)\right)=\frac{1}{2}\frac{\rho}{m}v^2(t)C_D\left(\mathrm{Ma}\right)S\ -av(t)\\y\left(t\right)=\left(\begin{matrix}r\left(t\right)\\v\left(t\right)\\\end{matrix}\right)\\\end{matrix}\right.
\end{equation}
\begin{table}[!h]
\centering
  \caption{Neural Controller Architecture}
\begin{tabular}{l|c}
\hline
	Layer &     No. Neurons/Activation Function \\\hline
Input layer	& 6/ Linear  $ [r^T(t), v^T(t)]^T$ \\ 
$1^{st}$ layer	&   40/Tanh \\ 
    $2^{nd}$ layer &  40/Sigmoid\\ 
	$3^{rd}$layer  &   40/Tanh \\ 
	Output layer  &  3/ Linear\\ 
 \hline
\end{tabular}
    \label{tab2}
\end{table}
\subsubsection{\color{blue} Step 3: Error Dynamical System}
By using the closed-loop reference model and the new mathematical form of the Lander dynamical motion, we can now construct the error dynamical system. Let $\delta  r=r-\hat{r}$,  $\delta  v=v-\hat{v}$ and $\delta  u=u-\hat{u}$  then
\begin{equation}
    \left\{\begin{matrix}\left(\begin{matrix}\frac{d\delta r}{dt}\\\frac{d\delta  v}{dt}\\\end{matrix}\right)=\left[\begin{matrix}\delta  v\\q\delta  v\left(t\right)+\delta u\left(t\right)+\Delta_{\delta}\left(v\left(t\right)\right)\\\end{matrix}\right]\\\ \Delta_{\delta}\left(v(t)\right)=\frac{1}{2}\frac{\rho}{m}v^2(t)C_D\left(\mathrm{Ma}\right)S\ -qv(t)\\z\left(t\right)=\left(\begin{matrix}\delta r\left(t\right)\\\ \delta v\left(t\right)\\\end{matrix}\right)\\\end{matrix}\right.
\end{equation}
According to lemma \ref{lemma4}, the approximation error $\delta u=u-\hat{u}$ given in the equation above can be expressed as: 
\begin{equation}
    \delta u\left(t\right)=\Lambda(s) \zeta(t)+\frac{\partial \alpha}{\partial \hat{v}}(0)\hat{v}(t)+\epsilon (t),
 \end{equation}
 where in this case we don't have a reference signal i.e. $d(t)=0$, the LPV matrix $ \Lambda(s(t))=\nabla_{s } u(\hat{s})$ is the Jacobian matrix 
 of the neural network at point $s \in \mathbb{Co}([r^T~ v^T]^T, [\hat{r}^T~ \hat{v}^T]^T)$, $\epsilon (t)=\mathcal{O}(\hat{v}(t))$ and  the training error$\alpha\left(\hat{v}\left(t\right)\right)=u\left(\hat{v}\left(t\right)\right)-\hat{u}\ \left(\hat{v}\left(t\right)\right)$. Then, the dynamical error became: 
\begin{eqnarray}
    \left\{ \begin{array}{ll}
    \begin{matrix}\left[\begin{matrix}\frac{d\delta r}{dt}(t)\\\frac{d \delta v}{dt}(t)\\\end{matrix}\right]&=\left[\begin{matrix}\delta v(t)\\q\delta v\left(t\right)+\Delta_{\delta}\left(v\left(t\right)\right)+\Delta_\epsilon\left(\hat{v}\left(t\right)\right)+\frac{\partial \alpha}{\partial \hat{v}}(0)\hat{v}(t)\\\end{matrix}\right] +B\Lambda\left(s\right)\left[\begin{matrix}\delta r\left(t\right)\\\delta  v\left(t\right)\\\end{matrix}\right],\\\\\Delta_{\delta}\left(v(t)\right)&=\frac{1}{2}\frac{\rho}{m}v^2(t)C_D\left(\mathrm{Ma}\right)S\ -qv(t),~\epsilon\left(\hat{v}\left(t\right)\right)=\mathcal{O}(\hat{v}(t))\\ z\left(t\right)&=C\left(\begin{matrix}\delta r\left(t\right)\\ \delta v\left(t\right)\\\end{matrix}\right),\ \ B=\left[\begin{matrix}\begin{matrix}\begin{matrix}0&0&0\\\end{matrix}\\\begin{matrix}0&0&0\\\end{matrix}\\\begin{matrix}0&0&0\\\end{matrix}\\\end{matrix}\\\begin{matrix}\begin{matrix}1&0&0\\\end{matrix}\\\begin{matrix}0&1&0\\\end{matrix}\\\begin{matrix}0&0&1\\\end{matrix}\\\end{matrix}\\\end{matrix}\right]\\\end{matrix}
    \end{array} \right.
\end{eqnarray}
\subsubsection{\color{blue} Step 4: IQCs Characterisation of the nonlinearities $\Delta_{\delta}$ and $\Delta_\epsilon$} 
To choose the type of IQCs of the uncertainties considered for this  example, let start with the static nonlinearity $\Delta_{\delta}\left(v(t)\right)=\frac{1}{2}\frac{\rho}{m}v^2(t)C_D\left(\mathrm{Ma}\right)S\ -qv(t)$. It is clear from the drawn Figure \ref{delt} that the uncertainty $\Delta_{\delta}\left(v(t)\right)$ is slope-restricted and sector bounded. Figure \ref{delt} provides the IQCs characterisation of this nonlinearity and plots the graph  $\Delta_{\delta}\left(v(t)\right)$ where $\Delta_{\delta}\left(v\left(t\right)\right)=\varphi\left(v,t\right)$ and $\tilde{\alpha}v\left(t\right)\le\varphi\left(v,t\right)\le\tilde{\beta} v\left(t\right)$ for all $v$.  Simple calculation of the slope gives $\tilde{\alpha}=\mathbf{diag} (-0.003,0.013,0.006)$ and $\beta=\mathbf{diag}( 0.008,-0.004,-0.004)$  representing the bounds of the uncertainty $\Delta_{\delta}.$

Similarly, we have the nonlinearity $\epsilon\left(\hat{v}\left(t\right)\right)=u\left(\hat{v}\left(t\right)\right)-\hat{u}\ \left(\hat{v}(t)\right)$. The graph of $\epsilon=\hat{f}\left(\hat{v}\left(t\right)\right)$ is shown in Figure \ref{epsi} where $\hat{\alpha}\hat{v}\left(t\right)\le\hat{f}\left(\hat{v}\left(t\right)\right)\le\hat{\beta}\hat{v}\left(t\right).$ The IQCs characterisation of this nonlinearity is thus conducted to  in $\hat{\beta}=\mathbf{diag} (-0.052,0,0.0053)$ and $ \widehat{ \alpha}=\mathbf{diag}(0.022,-0.024,-0.0043)$.   
\begin{figure}[thpb]
      \centering
   \includegraphics[scale=0.25]{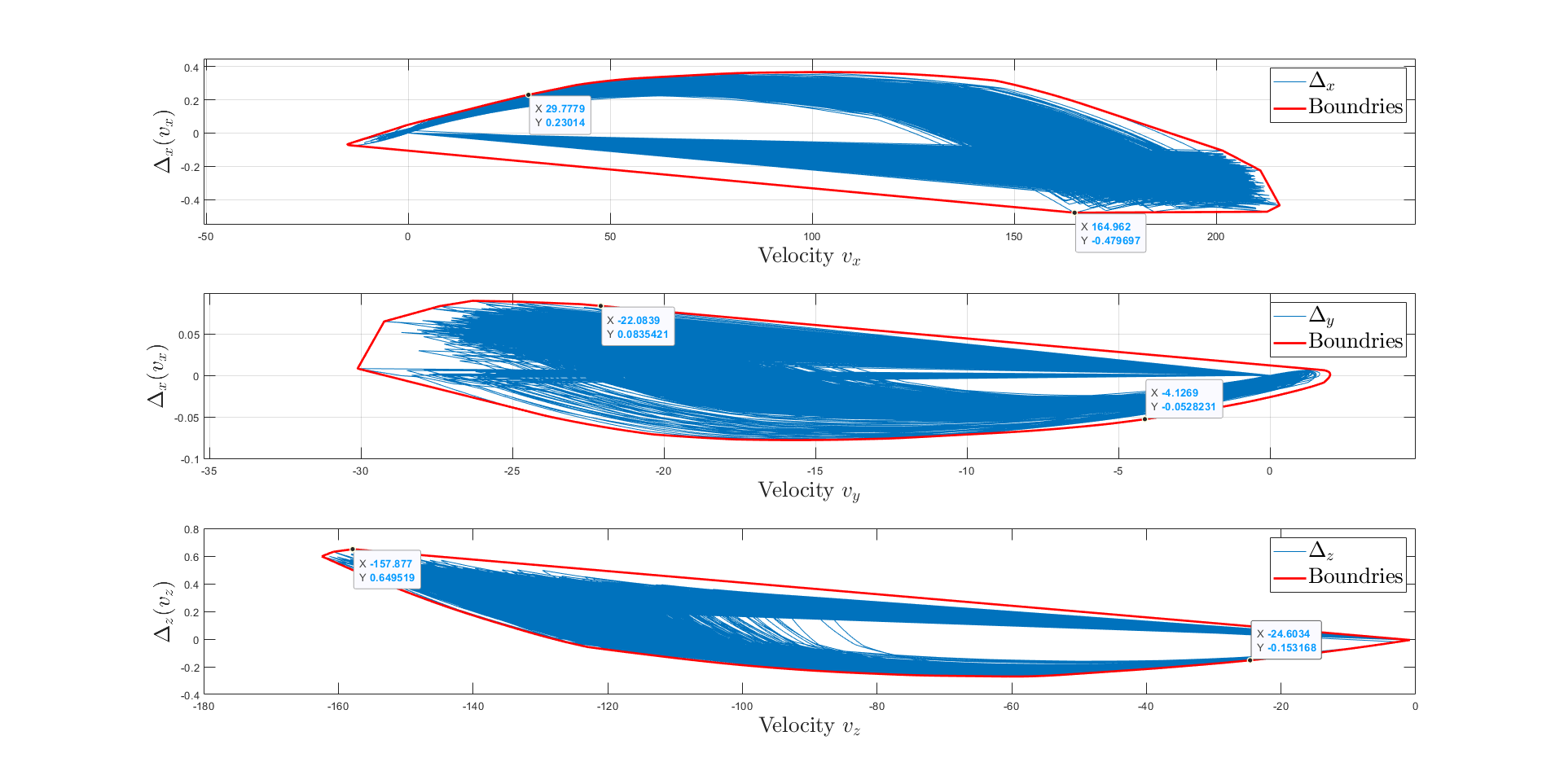}
\caption{ IQCs characterisation  of the nonlinearity $\Delta_{\delta}$.}\label{delt}
   \end{figure}

   \begin{figure}[thpb]
      \centering
   \includegraphics[scale=0.25]{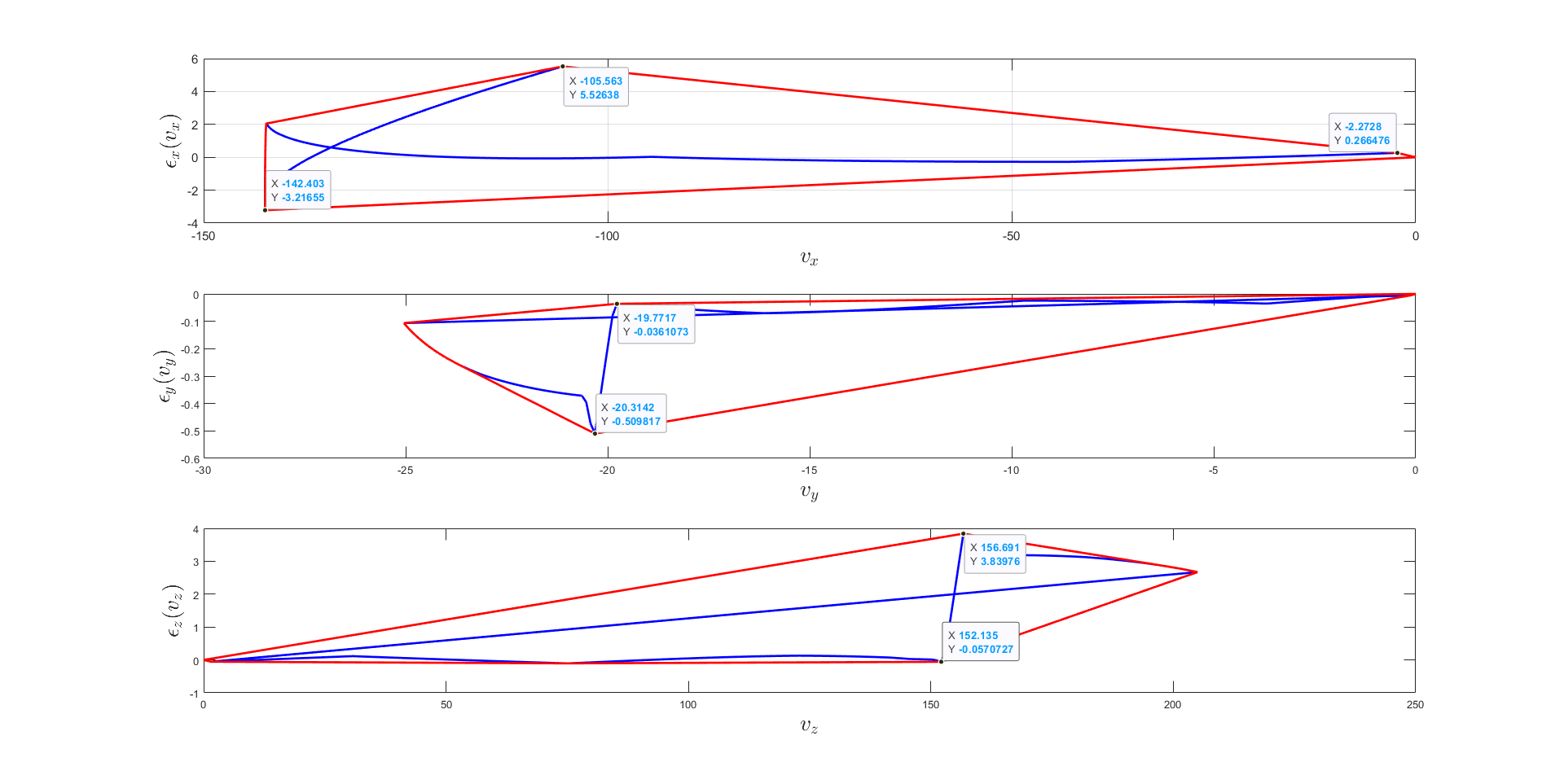}
\caption{ IQCs characterisation  of the nonlinearity $\Delta_\epsilon$.}\label{epsi}
   \end{figure}
\subsubsection{\color{blue} Step 5: IQCs based Robustness Analysis Results}
  According to the results given by IQClab Toolbox for the model corresponding to the original Apollo Robustness Control Problem: 
  \begin{eqnarray*}
      \left[\begin{matrix}\begin{matrix}\ \gamma_x&\gamma_y&\gamma_z\end{matrix} &\begin{matrix}\gamma_{v_x}&\gamma_{v_y}&\gamma_{v_z}\\\end{matrix}\\\end{matrix}\right ]^T=\left[\begin{matrix}\begin{matrix}0.0988&0.1992&0.1581\\\end{matrix}&\begin{matrix}0.0119&0.0244&0.0158\\\end{matrix}\\\end{matrix}\right]^T
  \end{eqnarray*}
The worst relative error between the closed-loop system with the trained neural controller and the closed-loop reference system and is less than 
$\left[\begin{matrix}\begin{matrix}9.88\%&19.92\%&15.81\%\end{matrix}&\begin{matrix}1.19\%&2.44\%&1.58\%\\\end{matrix}\\\end{matrix}\right]^T$ of the reference model outputs. 
To further illustrate the results of our controlled Lander with our designed neural controller, we perform some simulation results. Figure \ref{land} and Figure \ref{noland} show the response of the system with the trained NN controller for different two scenarios with different initial position and 
velocity. Clearly, the output remains inside the expected bound region. 
 \begin{figure*}[thpb]
      \centering
   \includegraphics[scale=0.3]{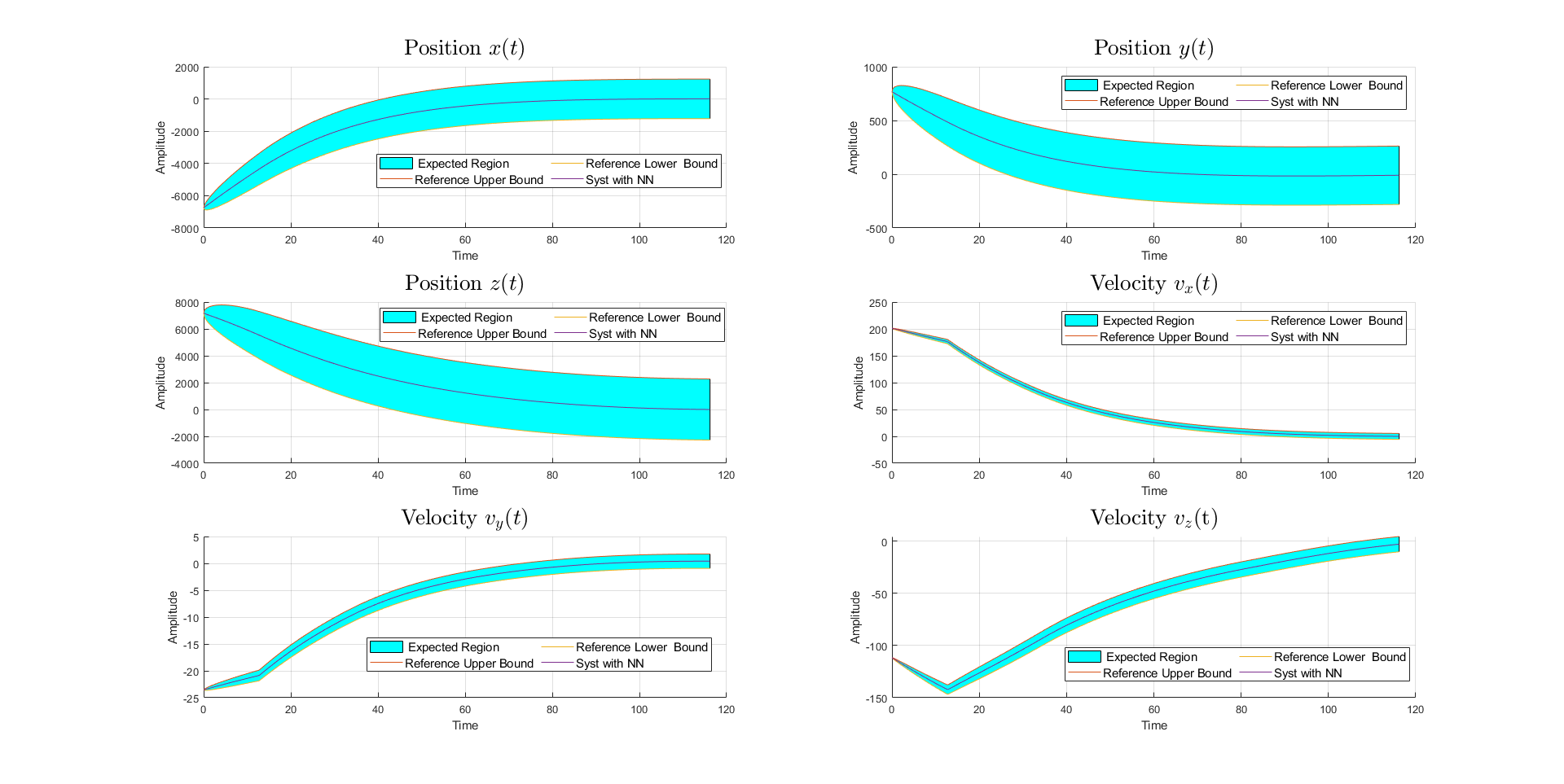}
\caption{Simulation Results, Successful Landing Scenario.}\label{land}
   \end{figure*}

    \begin{figure*}[thpb]
      \centering
   \includegraphics[scale=0.3]{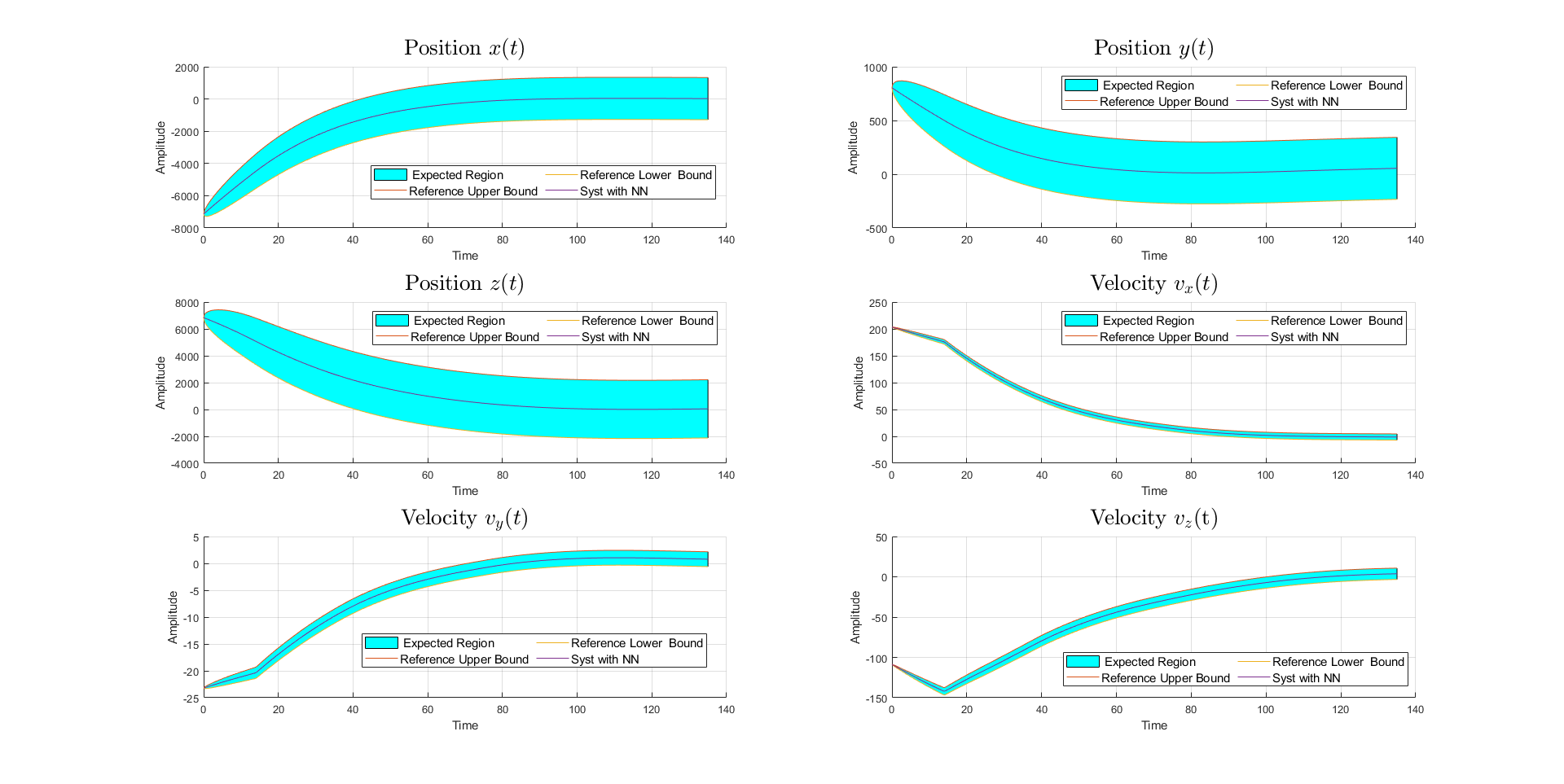}
\caption{Simulation Results, Failed Landing Scenario.}\label{noland}
   \end{figure*}
Clearly, the outcomes of this novel approach can be leveraged to enhance the robustness of neural network (NN) controllers by minimising dynamical error, thereby ensuring improved performance upon implementation. The estimated worst-case measure serves as a crucial indicator of the potential occurrence of undesirable behaviour within the closed-loop system when employing the designed NN controller.  

Furthermore, the error envelope can be effectively narrowed by utilising the Integral Quadratic Constraint (IQC) toolbox, which is compatible with both Linear Parameter-Varying (LPV) systems and hybrid systems. This compatibility is particularly significant due to the discontinuities introduced by activation functions such as ReLU (Rectified Linear Unit) and Leaky-ReLU, which result in a hybrid system framework when the system’s Jacobian is employed. By integrating IQC-based analysis, it becomes possible to systematically assess and constrain the uncertainties inherent in neural network-controlled systems, thereby improving overall stability and reliability.

\section{Conclusion}\label{sec5}
A novel technique, termed "keep-close," was introduced to ensure the safety and robustness of feedback systems equipped with neural network (NN) controllers amidst various types of uncertainties. The theoretical findings provided for worst-case analysis offer users preliminary insights into the maximum expected error between the reference model and the system's state when controlled by an NN amidst perturbations. The central concept behind this new methodology is to preserve the closed-loop system's output under an NN controller close to that of a robustly performing reference model, even when its input varies within a predefined bounded set. The "keep-close" strategy furnishes users with advanced knowledge regarding the worst-case rise and steady-state error (SSE) between the reference closed-loop model and the uncertain system managed by the neural controller. This insight is crucial for validating AI-based machine learning (ML) techniques through robust control theory, thereby enhancing user trust in these technologies. The efficacy and adaptability of this innovative approach have been empirically validated through practical implementation and testing on two distinct challenges: the Single-Link Robot-Arm Control Problem and the Deep Guidance and Control of the Apollo Lander. These real-world experiments have demonstrated the proposed method's effectiveness and flexibility across different application areas, offering empirical proof of its utility and adaptability.
\section{Appendix}
\label{appendix}
\begin{proof}
Firstly, let recall the definition of variables:
\[
\eta(t) = \begin{bmatrix}
d(t) \\ \hat{x}(t)
\end{bmatrix}, \quad
q(t) = \begin{bmatrix}
\delta(t) \\ \epsilon(t)
\end{bmatrix}.
\]
Now, we start with the error dynamic equation in (\ref{erroriqc-a}):
\begin{equation}\label{erroriqc_original}
\dot{\zeta}(t) = A\zeta(t) + B\mu(t) + \tilde{B}\delta(t),
\end{equation}
where $\zeta(0) = 0$. Also, the controller error \(\mu(t)\) from (\ref{erroriqccont}) is given by:
\begin{equation}\label{mu_expression}
\mu(t) = \Lambda(s)\zeta(t) + \frac{\partial \alpha}{\partial \hat{x}}(0,0)\hat{x}(t) 
+ \frac{\partial \alpha}{\partial d}(0,0)d(t) + \epsilon(t).
\end{equation}
 Substitute \(\mu(t)\) into \(\dot{\zeta}(t)\)
Substituting \(\mu(t)\) from (\ref{mu_expression}) into the first equation of (\ref{erroriqc_original}), we have:
\begin{align}
\dot{\zeta}(t) &= A\zeta(t) + B\Big[\Lambda(s)\zeta(t) 
+ \frac{\partial \alpha}{\partial \hat{x}}(0,0)\hat{x}(t) 
+ \frac{\partial \alpha}{\partial d}(0,0)d(t) + \epsilon(t)\Big] + \tilde{B}\delta(t), \\
&= \Big(A + B\Lambda(s)\Big)\zeta(t) 
+ B\begin{bmatrix}
\frac{\partial \alpha}{\partial d}(0,0) & \frac{\partial \alpha}{\partial \hat{x}}(0,0)
\end{bmatrix}
\begin{bmatrix}
d(t) \\ \hat{x}(t)
\end{bmatrix}
+ B\epsilon(t) + \tilde{B}\delta(t), \\
&= \Big(A + B\Lambda(s)\Big)\zeta(t) 
+ \begin{bmatrix}
B\frac{\partial \alpha}{\partial d}(0,0) & B\frac{\partial \alpha}{\partial \hat{x}}(0,0)
\end{bmatrix}
\begin{bmatrix}
d(t) \\ \hat{x}(t)
\end{bmatrix}
+ \begin{bmatrix}
\tilde{B} & B
\end{bmatrix}
\begin{bmatrix}
\delta(t) \\ \epsilon(t)
\end{bmatrix},\\
&= \Big(A + B\Lambda(s)\Big)\zeta(t) 
+ \begin{bmatrix}
B\frac{\partial \alpha}{\partial d}(0,0) & B\frac{\partial \alpha}{\partial \hat{x}}(0,0)
\end{bmatrix}
\eta(t)
+ \begin{bmatrix}
\tilde{B} & B
\end{bmatrix}
q(t). \label{zeta}
\end{align}
Now,  we move to  the output in (\ref{erroriqc-b}) which is given by: 
\begin{eqnarray}\label{output}
    z(t) &= C\zeta(t) + D\mu(t) + \tilde{D} \delta(t),
\end{eqnarray}
Similarly, substituting \(\mu(t)\) from (\ref{mu_expression}) into \(z(t)\) in  (\ref{output}), we have:
\begin{align}
z(t) &= C\zeta(t) + D\Big[\Lambda(s)\zeta(t) 
+ \frac{\partial \alpha}{\partial \hat{x}}(0,0)\hat{x}(t) 
+ \frac{\partial \alpha}{\partial d}(0,0)d(t) + \epsilon(t)\Big] + \tilde{D}\delta(t), \\
&= \Big(C + D\Lambda(s)\Big)\zeta(t) 
+ D\begin{bmatrix}
\frac{\partial \alpha}{\partial d}(0,0) & \frac{\partial \alpha}{\partial \hat{x}}(0,0)
\end{bmatrix}
\begin{bmatrix}
d(t) \\ \hat{x}(t)
\end{bmatrix}
+ D\epsilon(t) + \tilde{D}\delta(t), \\
&= \Big(C + D\Lambda(s)\Big)\zeta(t) 
+ \begin{bmatrix}
D\frac{\partial \alpha}{\partial d}(0,0) & D\frac{\partial \alpha}{\partial \hat{x}}(0,0)
\end{bmatrix}
\begin{bmatrix}
d(t) \\ \hat{x}(t)
\end{bmatrix}
+ \begin{bmatrix}
\tilde{D} & D
\end{bmatrix}
\begin{bmatrix}
\delta(t) \\ \epsilon(t)
\end{bmatrix},
\\
&= \Big(C + D\Lambda(s)\Big)\zeta(t) 
+ \begin{bmatrix}
D\frac{\partial \alpha}{\partial d}(0,0) & D\frac{\partial \alpha}{\partial \hat{x}}(0,0)
\end{bmatrix}
\eta(t)
+ \begin{bmatrix}
\tilde{D} & D
\end{bmatrix}
q(t).\label{output1}
\end{align}
Using  (\ref{zeta}) and (\ref{output1}), we can rewrite the system (\ref{erroriqc}) and the new system becomes:
\begin{equation}\label{new_system}
\Sigma_{\tilde{F}(\tilde{P},\Delta_{\delta})}^{\eta, q}:=  
\begin{cases} 
\dot{\zeta}(t) = \Big(A + B\Lambda(s)\Big)\zeta(t) + \begin{bmatrix}
B\frac{\partial \alpha}{\partial d}(0,0) & B\frac{\partial \alpha}{\partial \hat{x}}(0,0)
\end{bmatrix}\eta(t) + \begin{bmatrix}
\tilde{B} & B
\end{bmatrix}q(t), \\
z(t) = \Big(C + D\Lambda(s)\Big)\zeta(t) + \begin{bmatrix}
D\frac{\partial \alpha}{\partial d}(0,0) & D\frac{\partial \alpha}{\partial \hat{x}}(0,0)
\end{bmatrix}\eta(t) + \begin{bmatrix}
\tilde{D} & D
\end{bmatrix}q(t), \\
\zeta(0) = 0.
\end{cases}
\end{equation}
Now will try to eliminate the input of the uncertainties $\Delta_\delta(.)$ and $\Delta_\epsilon(.)$ from ‘virtual’ filters in (\ref{delta}) and (\ref{epsilon}),  respectively. From (\ref{delta}),  the ‘virtual’  filter equation is given by 
\begin{equation}
\begin{cases}
\dot{\psi}_{\delta}(t) = A_{{\delta}}\psi_{\delta}(t) + B_{{\delta}1} \delta(t) + B_{{\delta}2}y(t), \\
r_{\delta}(t) = C_{{\delta}}\psi_{\delta}(t) + D_{{\delta}1}\delta(t) + D_{{\delta}2}y(t), \\
\psi_{\delta}(0) = 0,
\end{cases}
\end{equation}
Knowing that \(y(t) = z(t) + \hat{y}(t)\), therfore, substituting \(y(t)\) into the system, we get:
\begin{align}
\dot{\psi}_{\delta}(t) &= A_{{\delta}}\psi_{\delta}(t) + B_{{\delta}1}\delta(t) + B_{{\delta}2}\big(z(t) + \hat{y}(t)\big), \\
r_{\delta}(t) &= C_{{\delta}}\psi_{\delta}(t) + D_{{\delta}1}\delta(t) + D_{{\delta}2}\big(z(t) + \hat{y}(t)\big).
\end{align}
From the earlier results in  (\ref{new_system}), \(z(t)\) is given by:
\begin{equation}
z(t) = \Big(C + D\Lambda(s)\Big)\zeta(t) + 
\begin{bmatrix}
D\frac{\partial \alpha}{\partial d}(0,0) & D\frac{\partial \alpha}{\partial \hat{x}}(0,0)
\end{bmatrix}
\eta(t) +
\begin{bmatrix}
\tilde{D} & D
\end{bmatrix}
q(t),
\end{equation}
and from the reference system \(\Sigma_{P}^{\pi^*}\) in (\ref{refcon}), \(\hat{y}(t)\) is:
\begin{equation}
\hat{y}(t) = C_r\hat{x}(t) + D_r d(t).
\end{equation}
Substitute \(z(t)\) and \(\hat{y}(t)\) into \(\dot{\psi}_{\delta}(t)\):
\begin{align}
\dot{\psi}_{\delta}(t) &= A_{{\delta}}\psi_{\delta}(t) + B_{{\delta}1}\delta(t) 
+ B_{{\delta}2}\Big(C + D\Lambda(s)\Big)\zeta(t) \nonumber \\
&\quad + 
\begin{bmatrix}
D\frac{\partial \alpha}{\partial d}(0,0) + B_{{\delta}2}D_r & D\frac{\partial \alpha}{\partial \hat{x}}(0,0) + B_{{\delta}2}C_r
\end{bmatrix}\eta(t) \nonumber \\
&\quad + 
\begin{bmatrix}
B_{{\delta}2}\tilde{D} & B_{{\delta}2} D
\end{bmatrix}q(t).
\end{align}
Similarly, substitute \(z(t)\) and \(\hat{y}(t)\) into \(r_{\delta}(t)\):
\begin{align}
r_{\delta}(t) &= C_{{\delta}}\psi_{\delta}(t) + D_{{\delta}1}\delta (t) 
+ D_{{\delta}2}\Big(C + D\Lambda(s)\Big)\zeta(t) \nonumber \\
&\quad + 
\begin{bmatrix}
D\frac{\partial \alpha}{\partial d}(0,0) + D_{{\delta}2}D_r & D\frac{\partial \alpha}{\partial \hat{x}}(0,0) + D_{{\delta}2}C_r
\end{bmatrix}\eta(t) \nonumber \\
&\quad + 
\begin{bmatrix}
D_{{\delta}2}\tilde{D} & D_{{\delta}2}D
\end{bmatrix}q(t).
\end{align}
Thus, the virtual filter for the uncertainty $\Delta_\delta$  is:
\begin{equation}\label{delta2}
\begin{cases}
\dot{\psi}_{\delta}(t) = A_{{\delta}}\psi_{\delta}(t) 
+ B_{{\delta}2}\Big(C + D\Lambda(s)\Big)\zeta(t) \\
\quad + 
\begin{bmatrix}
B_{{\delta}2}D\frac{\partial \alpha}{\partial d}(0,0) + B_{{\delta}2}D_r & B_{{\delta}2}D\frac{\partial \alpha}{\partial \hat{x}}(0,0) + B_{{\delta}2}C_r
\end{bmatrix}\eta(t) \\
\quad + 
\begin{bmatrix}
B_{{\delta}2}\tilde{D} + B_{{\delta}1}& B_{{\delta}2}D
\end{bmatrix}q(t), \\
r_{\delta}(t) = C_{{\delta}}\psi_\delta(t)
+ D_{{\delta}2}\Big(C + D\Lambda(s)\Big)\zeta(t) \\
\quad + 
\begin{bmatrix}
D_{{\delta}2}D\frac{\partial \alpha}{\partial d}(0,0) + D_{{\delta}2}D_r & D_{{\delta}2}D\frac{\partial \alpha}{\partial \hat{x}}(0,0) + D_{{\delta}2}C_r
\end{bmatrix}\eta(t) \\
\quad + 
\begin{bmatrix}
D_{{\delta}2}\tilde{D}+D_{{\delta}1} & D_{{\delta}2}D
\end{bmatrix}q(t), \\
\psi_{\delta}(0) = 0.
\end{cases}
\end{equation}
For the virtual filter for the uncertainty $\Delta_\epsilon $ in (\ref{epsilon}) is given  by: 
\begin{equation}\label{epsilon2}
\begin{cases}
\dot{\psi}_\epsilon(t) = A_{\epsilon}\psi_\epsilon(t) + B_{\epsilon1}\epsilon(t) + B_{\epsilon2}\eta(t), \\
r_\epsilon(t) = C_{\epsilon}\psi_\epsilon(t) + D_{\epsilon1} \epsilon(t) + D_{\epsilon2}\eta(t), \\
\psi_\epsilon(0) = 0,
\end{cases}
\end{equation}
Using an extended state $\chi(t) = [\zeta(t)~ \psi_{\delta}~\psi_{\epsilon}]^T$ and the equations (\ref{new_system}), (\ref{delta2}) and (\ref{epsilon2}) we get the extended system in (\ref{extendedsystem}).
\end{proof}
\section*{Acknowledgment}
The authors express their gratitude to Spin-Works for their support and to Tiago Amaral for providing the Apollo software used in data generation. Additionally, we extend our appreciation to Joost Veenman from SENER Aeroespacial, Tres Cantos, Spain, for granting access to the IQClab Toolbox, which was instrumental in this research. Furthermore, we sincerely thank Samir Bennani and David Sanchez de la Llana from the European Space Research and Technology Centre (ESA), 2201 AZ Noordwijk, The Netherlands, for their valuable feedback and support.

\printendnotes

\bibliography{sample}

\begin{biography}[example-image-1x1]{A.~One}
Please check with the journal's author guidelines whether author biographies are required. They are usually only included for review-type articles, and typically require photos and brief biographies (up to 75 words) for each author.
\bigskip
\bigskip
\end{biography}

\graphicalabstract{example-image-1x1}{Please check the journal's author guildines for whether a graphical abstract, key points, new findings, or other items are required for display in the Table of Contents.}

\end{document}